\documentclass[journal]{IEEEtran}

\date{}
\usepackage{amsfonts}
\usepackage{amsmath, amsthm, amssymb}
\usepackage{graphicx}
\usepackage{hyperref, comment, cite}
\usepackage{bbm, amssymb}
\usepackage{mathrsfs} 
\usepackage{subfig}

\usepackage{tikz}

\usetikzlibrary{fit}			\usetikzlibrary{backgrounds}	

\newtheorem{theorem}{Theorem}
\newtheorem{lemma}[theorem]{Lemma}

\newtheorem{corollary}[theorem]{Corollary}
\newtheorem{definition}[theorem]{Definition}
\newtheorem{example}[theorem]{Example}

\newtheorem{remark}[theorem]{Remark}
\DeclareMathOperator{\e}{e}
\DeclareMathOperator{\w}{w}
\usepackage[utf8]{inputenc}

\definecolor{darkblue}{rgb}{0, .07, .5}
\definecolor{darkred}{rgb}{0.5,0,0}
 \definecolor{mahogany}{rgb}{0.65, 0., 0.5}

\newcommand{\bxo}{\boldsymbol{x_1}}
\newcommand{\bxoo}{\boldsymbol{x_2}}
\newcommand{\bmu}{\boldsymbol{\mu}}
\newcommand{\blambda}{\boldsymbol{\lambda}}

\newcommand{\RN}[1]{%
  \textup{\uppercase\expandafter{\romannumeral#1}}%
}

\DeclareMathOperator*{\argmin}{arg\,min} 

\title{Transmission of a Bit over a Discrete Poisson Channel with Memory} 
\author{Niloufar~Ahmadypour and Amin~Gohari,~\IEEEmembership{Senior Member,~IEEE}
\thanks{This work was supported in part by INSF grant 96015883 and INSF grant on ``Nanonetwork Communications".

Niloufar Ahmadypour is with the Department of Electrical Engineering, Sharif University of Technology, Tehran, Iran (e-mail: ahmadypour\_n@ee.sharif.edu).
 
Amin Gohari is with the Tehran Institute for Advanced Studies (TeIAS),
Tehran, Iran (email: a.gohari@teias.institute).

A short version of this paper was presented at the 2020 IEEE Information Theory Workshop (ITW 2020).
}
}
\begin{document}
\maketitle
\begin{abstract}

A coding scheme for transmission of a bit  maps a given bit to a sequence of channel inputs (called the codeword associated with the transmitted bit). In this paper, we study the problem of designing the best code for a discrete Poisson channel with memory (under peak-power and  total-power constraints). The outputs of a discrete Poisson channel with memory are Poisson distributed random variables with a mean comprising of a fixed additive noise and a linear combination of past input symbols. Assuming a maximum-likelihood (ML) decoder, we search for a codebook that has the smallest possible error probability. This problem is challenging because  error probability of a code does not have a closed-form analytical expression. For the case of having only a total-power constraint,  the optimal code structure is obtained, provided that the blocklength is greater than the memory length of the channel. For the case of having only a peak-power constraint, the optimal code is derived for arbitrary memory and blocklength in the high-power regime. For the case of having both the peak-power and total-power constraints, the optimal code is derived for memoryless Poisson channels when both the total-power and the peak-power bounds are large.
\end{abstract}
\section{Introduction}

Discrete Poisson channels are widely used to model molecular and optical communication channels \cite{ref27, ref28, ref30, amini,rev22,rev21, lapidoth2009,lapidoth2011, cao,cao2}.  In particular, the Poisson distribution is used to model one of the main types of noise in molecular communication, namely the counting noise. Moreover, discrete Poisson channels with memory are also used to model molecular transmitters \cite[Sec. 3.6]{Bressloff}\cite{arjmandi1, mosayebi}.  Previous works in optical and molecular communication investigate the capacity of discrete Poisson channels (e.g. \cite{rev22,rev21, lapidoth2009,lapidoth2011, cao,cao2,amini}).  Also various modulation techniques are proposed for this channel \cite{jamali,mmm1,mmm2,mmm3}. For instance,  \cite{jamali} designs codes without any access to the channel state information.

Three key parameters in the problem of communicating a message over  a discrete Poisson channel are (i) the length of the message, (ii) the number of times the channel is used, and (iii) the error probability of the transmission. On the one hand, the classical notion of capacity assumes transmission of a long message of length $NR$ bits over $N$ uses of the channel with a vanishing probability of error. The receiver waits to receive the entire $N$ output symbols in order to decode the $NR$ message bits. On the other hand, practical  modulation schemes usually code a few bits of information over a few channel uses with a given probability of error in each transmission block; decoding  is  performed  over  smaller sequences of output symbols (the receiver can decode each transmission block separately). Unlike capacity-achieving codes, the error probability of a modulation scheme cannot converge to zero when the number of channel uses is limited.  In this paper, we consider the latter setting and study the  problem of finding an \emph{optimal}  coding  scheme to send a message of length one bit over a block of length $N$  with the lowest possible probability of error. Since we consider the problem of transmission of a bit $B\in\{1,2\}$, the receiver has to distinguish between two possibilities based on its received output sequence. 

Once the two codewords are designed and fixed, the problem at the receiver reduces to that of a hypothesis testing problem, the optimal decoder is an ML decoder, and the probability of error can be computed. We are interested to find the best choice of codewords that minimize the error probability of the ML decoder

The problem of communicating a bit using a fixed number of channel uses is of relevance for the following reason: firstly, it ensures that the transmitted bit is decoded after receiving a fixed number of output symbols; in order words, if we repeatedly use the coding  scheme to send multiple bits over consecutive blocks, the receiver can decode the transmitted bits one by one, with limited decoding delay and complexity.  Secondly, some applications in molecular communication require only the transmission of a few bits of information (low rate communication). In diffusion-based molecular communication, carriers of information are  molecules that physically travel from transmitters to receivers. The diffusion process is usually slow and low rate communication is of relevance in this context. Several applications of data transmission at low rates are reviewed in \cite{felice}. For instance, in targeted drug delivery applications there is a control node that orders another center to release the drug or stop it (a one-bit  message) \cite{tcplike}. Moreover, as argued in the literature, molecular transmitters  and receivers may be resource-limited  devices, and utilizing sophisticated coding schemes with long blocklengths may not be practically feasible in molecular communication. Therefore, considering codes with small blocklengths  (and therefore low rates) are of particular relevance to molecular communication. Very short blocklengths are also of relevance in certain wireless applications (see \cite{4ref4, ultrasmall} for a list of applications). 

The authors in \cite{taherzade} study the transmission of a single symbol over a memoryless additive white Gaussian noise (AWGN) channel where the decoding delay is assumed to be zero. The authors in \cite{ultrasmall,flip}, consider the transmission of up to two bits of information over memoryless BEC, BSC and Z-channels and derive the optimal codes of arbitrary blocklength (the problem is open for more than two bits of information). The key difficulty in finding the optimal codewords is that the exact error probability does not have a nice expression. It is shown in \cite{edc} that taking ``the minimum distance" of a code as a proxy for its error probability  can be misleading. Thus, one has to consider the exact structure of the codes, and cannot simply work with certain code parameters (as commonly adopted in coding theory). Similarly, inequalities on the error probability such as the one  by Gallager on the error probability \cite[Ex. 5.19]{gallager} are not immediately helpful.

In this paper, we consider the problem of communicating a bit over discrete Poisson channels. In other words, we consider the same problem as considered in \cite{ultrasmall} for discrete Poisson channels. However, our setup is different from that of  \cite{ultrasmall} as we also consider channels with memory. A challenge in the study of the Poisson channels is that we are using the optimal maximum-likelihood decoding at the receiver. It turns out that to follow the ML decoding rule for a Poisson channel with memory, the decoder needs to take a threshold and compare a weighted linear sum of its received sequence with that threshold. In other words, because of the memory of the channel, a transmission at a certain time slot will affect multiple received symbols at the receiver. This complicates the expression of the exact error probability which will also be in terms of Poisson tails that do not have explicit analytical closed forms. Since we are interested in the optimal codewords, an approximation of the Poisson tail with a Gaussian tail can be suboptimal. The problem is further complicated by the fact that the first derivative and second derivative conditions are complicated-looking expressions and are not easily amenable to analysis.  Furthermore, we show that the problem of finding the optimal code is a nonconvex optimization problem. Thus, to circumvent local minima, in one of the proofs we relax the optimization problem of finding the best code in such a way that the local minima are eliminated while the global minimum is preserved.

The main results of this paper are as follows:
\begin{enumerate}
\item 
For blocklengths larger than channel memory, we provide an optimal code under the total-power  constraint. A partial result is provided for the case of blocklengths less than the channel memory. 
\item 
For any arbitrary blocklength $N$ the optimal code is derived under the peak-power  constraint in the  high-power  regime. In this case, the strategy of setting input at its maximum possible value for one input message, and setting the input to zero for the other input message (on/off keying strategy) is shown to be optimal.

\item Under both  peak-power  and total-power  constraints, an optimal code is derived in the high-power  regime for Poisson channels without memory. 
\end{enumerate}

This paper is organized as follows.  In Section~\ref{three} we introduce the system model and formally state our problem. Section~\ref{four} presents our main results. In Section~\ref{five} we provide the proofs of our theorems and lemmas that help to prove theorems. Some of the lemmas and proofs are moved to appendices.

\section{Problem Formulation}
\label{three}
In this section, we describe the system model in detail. The section begins by giving the mathematical formulation of a discrete-time Poisson channel.   While we consider memoryless  channels as well as channels with memory, we do not assume output feedback from the receiver to the transmitter. Next, the transmitter and receiver models are defined. 

\textbf{Channel Model:} We begin by defining a Poisson channel without memory. A memoryless Poisson channel  takes an input $X\in\left[0,\infty\right)$ and outputs a symbol $Y\in\mathbb{N}\cup\{0\}$ where $Y\sim~\mathsf{Poisson}\left(X+d\right)$. Here, $d\in\left[0,\infty\right)$ is the dark background noise. In other words, the conditional distribution of output given input is as follows:
\begin{align*}
W\left(y \vert x\right)=e^{-\left(x+d\right)}\frac{{\left(x+d\right)}^y}{y!}.
\end{align*}
If a memoryless channel is used $N$ times, the input sequence $\boldsymbol{X}=\left[X_0, X_1, \ldots , X_{N-1}\right]$ is mapped to an output sequence $\boldsymbol{Y}=\left[Y_0, Y_1, \ldots , Y_{N-1}\right]$ where $$W\left(\boldsymbol{y}|\boldsymbol{x}\right)=\prod_{i=0}^{N-1} W\left(y_i|x_i\right).$$ A Poisson channel with memory is defined as follows (see \cite{amini, gohari2016information}): if the channel has memory of order $K$ , we associate the channel with a sequence  $$\boldsymbol{\pi}=\left[\pi_0, \pi_1, \ldots, \pi_{K-1}\right]$$ of length $K$. We call this sequence the channel coefficients.  To clarify the physical meaning of the channel coefficients $\pi_j$, consider a transmitter which is able to  release molecules into the environment every $T_s$ seconds. The transmitter can choose the number of released molecule at the beginning of each time-slot. The released molecules move randomly and diffuse into the environment, until they hit the receiver upon which they are absorbed. The channel coefficient $\pi_j$ indicates the probability that a  molecule released by the transmitter at the beginning of time slot $i$ hits the receiver during  the $i+j$-th time slot. { The sequence $\pi_i$ may sum up to a number less than one (as some released molecules may not hit the receiver surface at all). } We refer the readers to \cite[p. 9, Sec. II.D.2]{gohari2016information} for a more detailed explanation.  Assume that the channel is used $N$ times. Then, the input to the channel is a sequence $\left[X_0, X_1, \ldots , X_{N-1}\right]$. We assume that $X_i=0$ for $i<0$ or $i>N-1$ (meaning that transmission occurs only from time instance $0\leq i\leq N-1$). The output at time instance is given by \begin{equation}Y_i\sim \mathsf{Poisson}\left(d+\sum_{j=0}^{K-1}\pi_j X_{i-j}\right).\label{eqnA1}\end{equation}
 The physical meaning of the above equation can be found at 
\cite[p. 9, Sec. II.D.2]{gohari2016information}. { In Fig.\ref{channel}, for $N=3$ and $\boldsymbol{\pi}=(0.5,0.5)$, we sketch a Poisson channel with memory. Each output symbol follows a Poisson distribution whose mean is depicted in this figure.} Briefly speaking, $ \mathsf{Poisson}(X_{i-j})$ molecules are transmitted at the beginning of time slot $i-j$; a fraction of these molecules hit the receiver during the $i$-th time slot. This fraction of molecules is distributed according to $ \mathsf{Poisson}(\pi_jX_{i-j})$. It can be shown that the total number of received molecules equals $\mathsf{Poisson}\left(\sum_{j=0}^{K-1}\pi_j X_{i-j}\right) $ plus  the background noise distributed according to $\mathsf{Poisson}(d)$.

Observe that $Y_i\sim\mathsf{Poisson}\left(d\right)$ if $i<0$ or $i>N+K-2$. The output sequence $\boldsymbol{Y}=\left[Y_0, Y_1, \ldots , Y_{N+K-2}\right]$ for times $0\leq i\leq N+K-2$ has the following conditional distribution given the input sequence:
\begin{align*}
W&\left(\boldsymbol{y}|\boldsymbol{x}\right)\\
&=
\prod_{i=0}^{N+K-2}e^{-\left(d+ \sum_{j=0}^{K-1}\pi_j x_{i-j}\right)}\frac{{\left(d+ \sum_{j=0}^{K-1}\pi_j x_{i-j}\right)}^{y_i}}{y_i!}. 
\end{align*}
Observe that when $\boldsymbol{\pi}=\left[\pi_0, \pi_1, \ldots, \pi_{K-1}\right]=\left[\pi_0,0,0, \ldots , 0\right]$, the channel with memory   reduces to a memoryless channel.

\begin{figure}
 \includegraphics[width=.45\textwidth]{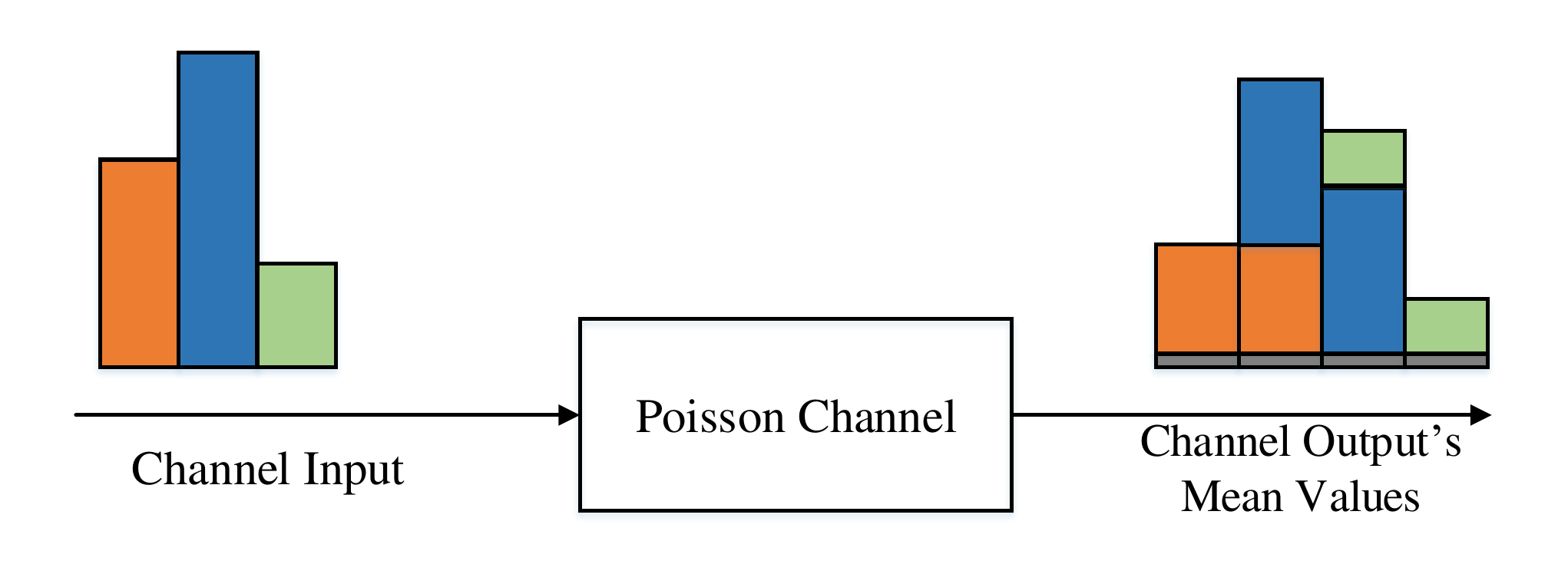}
 \centering
 \caption{A  Poisson channel with memory $\boldsymbol{\pi}=(0.5,0.5)$. Given an  input ($\boldsymbol{x}=(x_{0},x_1,x_2)$) of three non-negative real numbers, the outputs is a sequence of non-negative integers ($(Y_0, Y_1, Y_2, Y_3)$) where $Y_i$ has a Poisson distribution whose mean value is depicted. For instance, $Y_0$ has a Poisson distribution with mean $0.5x_0+d$, while $Y_1$ follows a Poisson distribution with mean $0.5x_0+0.5x_1+d$. The dark noise $d$ is illustrated in gray in this figure.  }
 \label{channel}
 \end{figure}

\textbf{Transmitter Model:} The transmitter has a uniform bit $B\in\{1,2\}$ that it wishes to communicate to the receiver.  For a blocklength $N$, the transmitter sends the codeword
$$\boldsymbol{x_1}=\left[x_{10}, x_{11}, \ldots , x_{1\left(N-1\right)}\right]$$
if $B=1$, or
$$\boldsymbol{x_2}=\left[x_{20}, x_{21}, \ldots , x_{2\left(N-1\right)}\right]$$
if $B=2$. We say that the codewords satisfy the peak-power  constraint $A$ if $x_{ji}\leq  A$ for $j\in\{1,2\}, i\in\{0,1, \ldots , N-1\}$. The codewords satisfy the  total-power  constraint $P$ if
\begin{align}
\label{jadid1}
 \sum_{i=0}^{N-1} x_{ji} \leq P, \qquad j=1,2.
\end{align}
In the context of molecular communication, the peak-power corresponds to the maximum number of molecules that could be potentially produced and released by the transmitter during each transmission time slot. However, the transmitter might not be able to maintain this maximum production during the entire transmission period. The total-power constraint corresponds to the total number of molecules that could be produced during the entire transmission period.

The following notation is used throughout the paper: we set
\begin{align} \lambda_i=\sum_{j=0}^{K-1}\pi_j x_{1\left(i-j\right)}=\left[\boldsymbol{x_1}* \boldsymbol{\pi}\right]\left(i\right), \qquad 0\leq i\leq N+K-2 \label{eqnA2}
\end{align}
and
\begin{align} \mu_i=\sum_{j=0}^{K-1}\pi_j x_{2\left(i-j\right)}=\left[\boldsymbol{x_2}* \boldsymbol{\pi}\right]\left(i\right), \qquad 0\leq i\leq N+K-2 \label{eqnA3}
\end{align}
where $*$ denotes the convolution operator. We also assume that $\lambda_i=\mu_i=0$ for $i<0$ or $i>N+K-2$. From \eqref{eqnA1}, for $B=1$ we have $Y_i\sim \mathsf{Poisson}\left(\lambda_i+d\right)$, and for $B=2$, we have $Y_i\sim \mathsf{Poisson}\left(\mu_i+d\right)$.

\textbf{Receiver Model:} Throughout the paper, we assume that the receiver uses the \emph{maximum-likelihood decoding} on the received sequence $\left[y_0, y_1,  \ldots , y_{N+K-2}\right]$ to produce an estimate of  the input bit $\hat B$. 
Given that each codewords has an equal a priori probability, the optimal receiver with minimum error probability is the ML receiver, and the decision rule (DR) is derived as follows: { using the values of $\lambda_i$ and $\mu_i$ from \eqref{eqnA2} and \eqref{eqnA3}, the probability of observing $(Y_0, Y_1, \cdots, Y_{N+K-2})$ when $B=1$ equals
$$\prod_{i=0}^{N+K-2}e^{-\left(d+ \lambda_i\right)}\frac{{\left(d+ \lambda_i\right)}^{Y_i}}{Y_i!},$$
while the probability of observing $(Y_0, Y_1, \cdots, Y_{N+K-2})$ when $B=2$ equals
$$\prod_{i=0}^{N+K-2}e^{-\left(d+ \mu_i\right)}\frac{{\left(d+ \mu_i\right)}^{Y_i}}{Y_i!}.$$
}
Therefore, we decode $\hat{B}=1$ if 
\begin{equation}
  \sum_{i=0}^{N+K-2} a_i\ Y_i \, \geq  \, b.\label{eqn:LMu}
\end{equation}
where $a_i=\log \frac{\lambda_i+d}{\mu_i+d}$ and $b={\sum_{i=0}^{N+K-2}} \left(\lambda_i-\mu_i\right)$.
We decode $\hat{B}=2$ if the right-hand side is greater than the left-hand side. The average  error probability is the probability that $\hat{B}$ is not equal to $B$. 
We are interested in optimal codewords $\boldsymbol{x_1}$ and $\boldsymbol{x_2}$ that minimize the  average  error probability under given total-power and/or  peak-power constraints  on the codewords. 

 \section{Main Results}
 \label{four}

In this section, we present our main results. This section is divided into three subsections which classify the results based on the assumptions made about total-power and peak-power constraints. 
 
\begin{figure*}[!t]
\centering
\subfloat[ $K=2$, $\boldsymbol{\pi}=(1/2,1/2)$.]{\includegraphics[width=3.25 in]{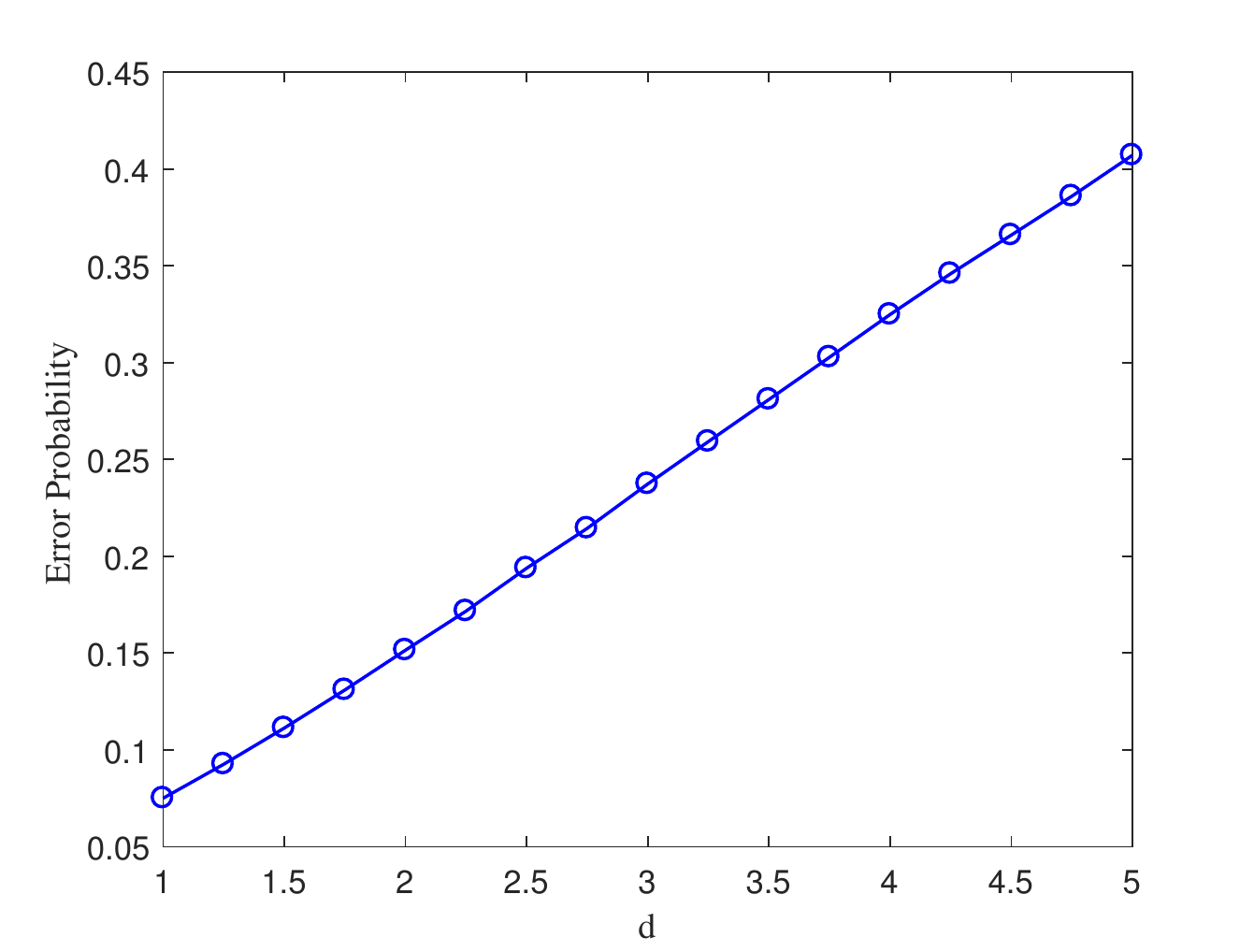}%
\label{figd_1}}
\hfil
\subfloat[ $K=1$, $\boldsymbol{\pi}=(1)$. ]{\includegraphics[width=3.25 in]{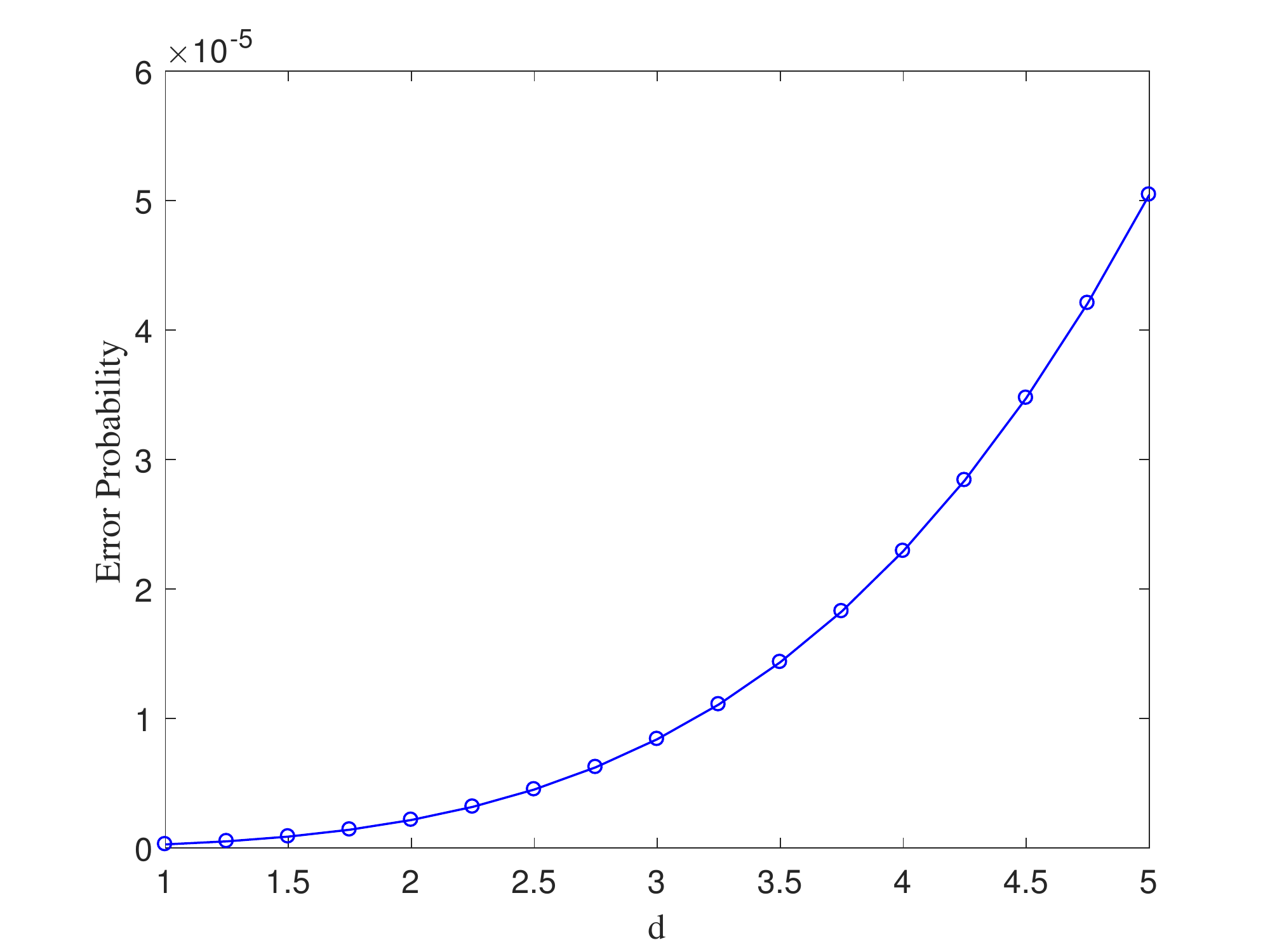}%
\label{figd_2}}
\caption{The error probability of the code given in \eqref{figdeq} as a function of $d$ for $d\in[1,5]$. The left sub-figure considers a Poisson channel with memory and the right sub-figure considers a Poisson channel without memory. Here, a total-power constraint $P=20$ is assumed. According to Theorem \ref{t1}, the code in \eqref{figdeq} is optimal for both channels for any value of $d$.}
\label{figd_sim}
\end{figure*}

\subsection{Code design subject to a  total-power  constraint}
In this section, we only consider a  total-power  constraint on the individual codewords ($P$ is finite but $A=\infty$). Our first result in Theorem~\ref{t1} identifies a particular structure for the codewords when the blocklength $N$ is strictly larger than the channel memory $K$. This structure can be suboptimal for blocklengths $N\leq K$. To illustrate this, we find the optimal codewords for the special case of $N=K=2$. Finding the optimal codewords  for blocklengths $N\leq K$ (in the general case) seems to be difficult.

\begin{theorem}
\label{t1}Consider a discrete Poisson channel with memory  of order $K$ and dark noise $d>0$. 
Assume that  blocklength $N$ is greater than or equal to $K+1$, and the two codewords $\boldsymbol{x_1}$ and $\boldsymbol{x_2}$ satisfy the  total-power constraint $P$ in \eqref{jadid1}. 
Then, the following codeword pair is optimal in the sense of minimizing the ML decoder error probability:
\begin{align}
\begin{pmatrix}
\boldsymbol{x_1}\\
\boldsymbol{x_2}
\end{pmatrix}= \Bigg( \overbrace{ \begin{matrix}
P & 0 &  \cdots  & 0  \\     
 0& 0& \cdots  &0         \end{matrix}}^K \quad
\overbrace{\begin{matrix}
   0   & 0&  \cdots &0\\
  P  &0& \cdots  & 0
\end{matrix}}^{N-K}
 \Bigg) \label{eqn1}
\end{align} 
\end{theorem}

We remark that the above theorem does not claim that the sequences given in \eqref{eqn1} is the \emph{only} optimal codeword pairs. It just claims that \eqref{eqn1} is an optimal codeword pair that minimizes the error probability of the ML receiver.

Proof of the above theorem is given in Section~\ref{five1} based on a suitable relaxation of the optimization problem of finding the best code. As shown  later in Remark~\ref{remarkN}, the problem of finding the optimal code is a nonconvex optimization problem, and the relaxation technique is used to circumvent local minima.

An interesting observation here is that the error probability for $N=K+1$ is the same as the error probability for any blocklength $N>K+1$, \emph{i.e.,} increasing $N$ beyond $K+1$ does not decrease the error probability.

\begin{figure*}[!t]
\centering
\subfloat[ $K=2$, $\boldsymbol{\pi}=(1/2,1/2)$ and $d=0.25$.]{\includegraphics[width=3.25 in]{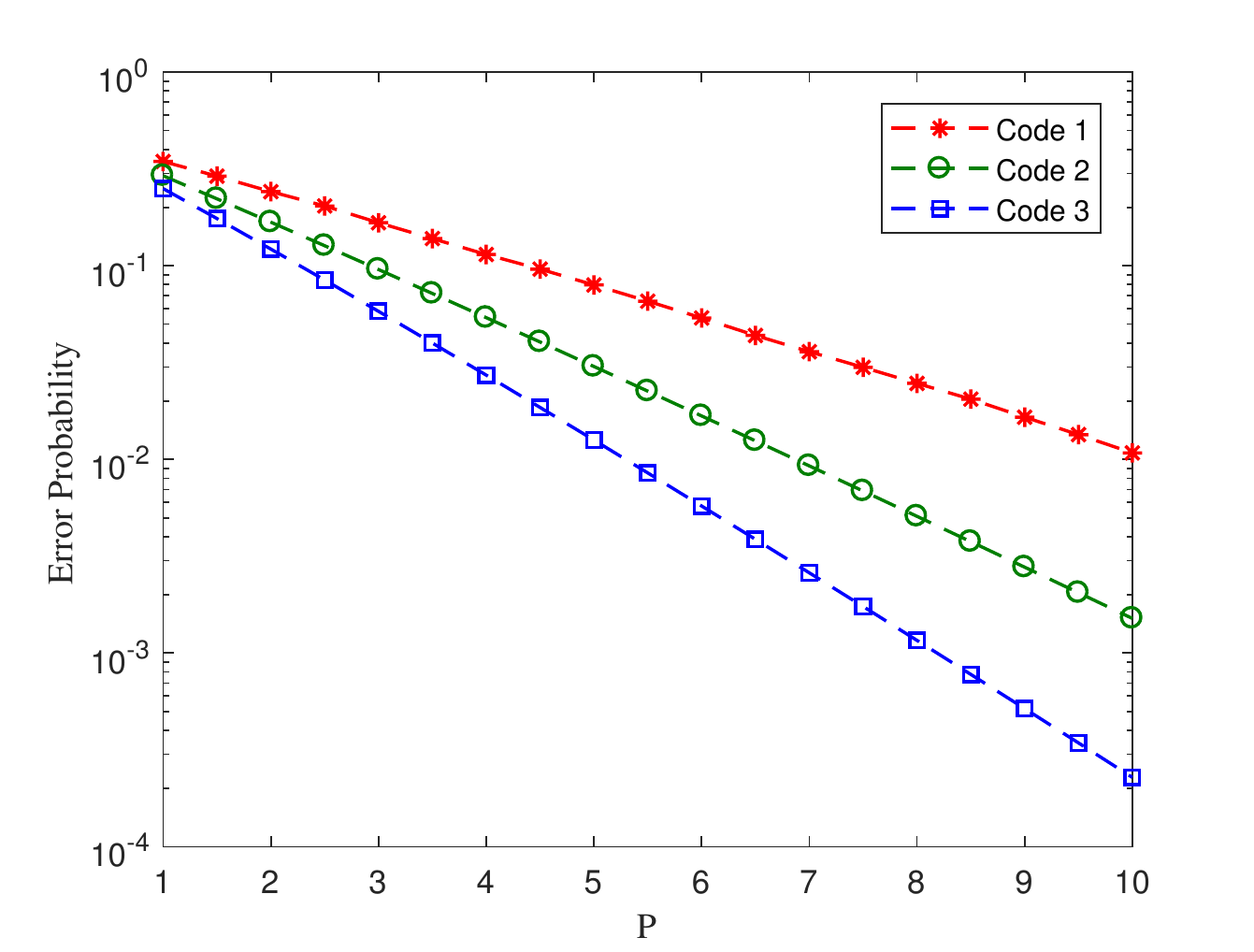}%
\label{figp_1}}
\hfil
\subfloat[ $K=1$, $\boldsymbol{\pi}=(1)$ and $d=0.25$. ]{\includegraphics[width=3.25 in]{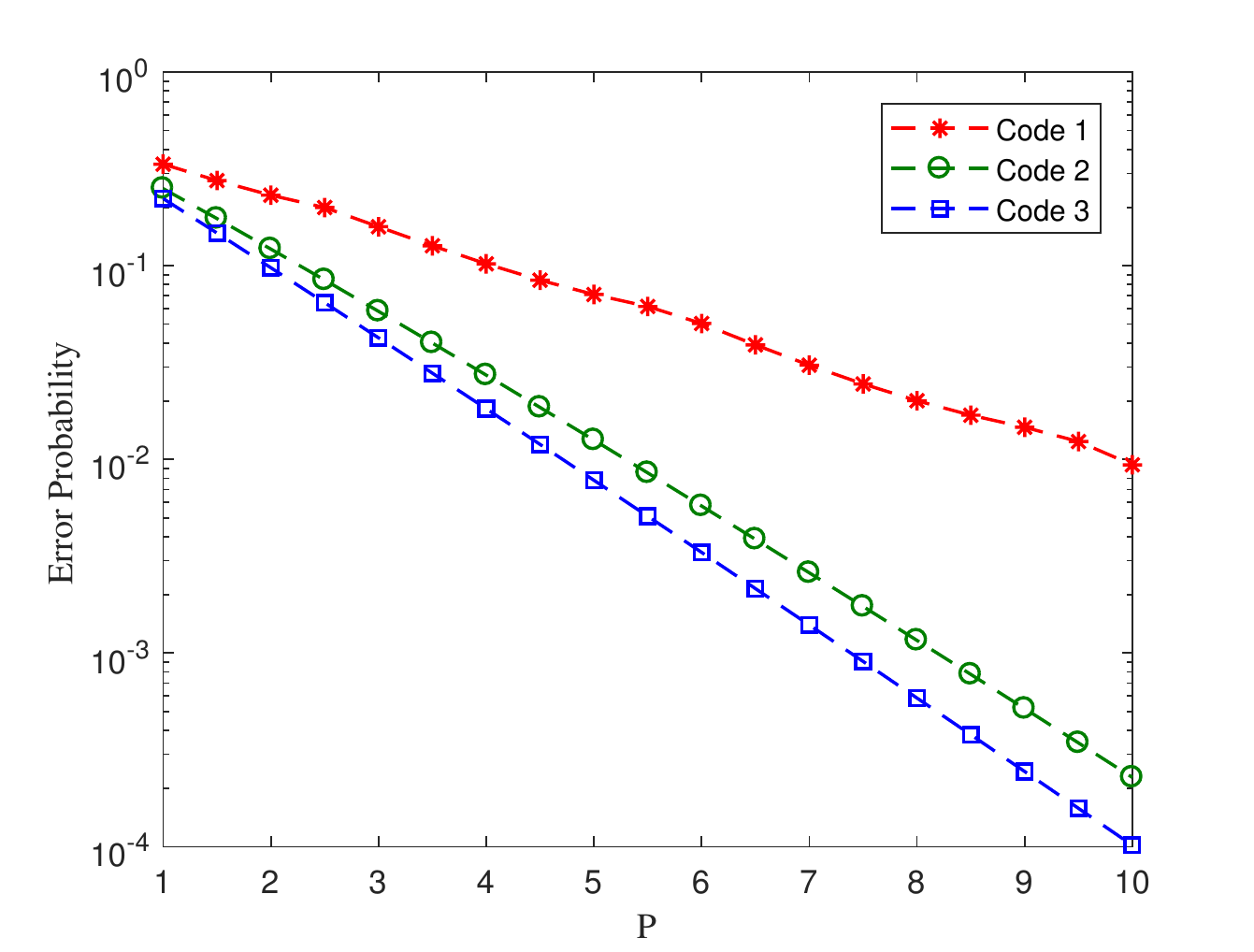}%
\label{figp_2}}
\caption{Comparison of the error probabilities of   three coding schemes given in \eqref{newp} for two different channel models. Code 3 is the optimal code given in Theorem \ref{t1}. Code 1 corresponds to the on/off keying strategy of sending a sequence at maximum power for one message, and sending nothing for the other message.}
\label{figp_sim}
\end{figure*}

\begin{corollary}
\label{c1} If the channel is without memory  ($K=1$), an  optimal codeword pair for blocklength  $N\geq 2$ is
\begin{align}
\Bigg(\overbrace{\begin{matrix}
  P & 0  & \cdots  & 0 \\
 0 & P & \cdots  & 0
\end{matrix}}^{N}\Bigg).
\label{col1}
\end{align}

Since the channel is without memory, any column-permutation of the construction in \eqref{col1} also yields an optimal code.

\end{corollary}

\textbf{Numerical Simulations:}
In Figure \ref{figd_sim}, we illustrate the error probability of ML-decoding for the following code:
\begin{align}
    C=\begin{pmatrix}
    \boldsymbol{x}_1\\
    \boldsymbol{x}_2
    \end{pmatrix}=
    \begin{pmatrix}
    P &0&0\\
    0&0&P
    \end{pmatrix},
    \label{figdeq}
\end{align}
for two channels, one with memory $K=2$ and the channel coefficients $\boldsymbol{\pi}=(1/2,1/2)$ (Fig.~\ref{figd_1}) and another one without memory (i.e. $\boldsymbol{\pi}=(1)$). 
This code is optimal for both channels used in this figure due to Theorem \ref{t1}. Firstly, observe that the channel memory has a severe impact on the error probability (in both cases we use the same total-power $P=20$). As expected, we observe that the error probability increases as we increase the background noise. This figure also shows that a channel without memory is more sensitive to an increase of the background noise from $d=1$ to $d=5$ (one curve is almost linear while the other curve is exponential).

In Figure~\ref{figp_sim} we compare the error probabilities of three different coding schemes for a channel without memory (Fig.~\ref{figp_2}) and for a channel with memory $\boldsymbol{\pi}=(1/2,1/2)$ (Fig.~\ref{figp_1}). In both case, the background dark noise is set to be $d=0.25$. The three codes used are as follows:
\begin{align}
\label{newp}
C_1&=\begin{pmatrix}
P/4 & P/4 &P/4&P/4\\
0&0&0&0
\end{pmatrix}, \nonumber\\
C_2&=\begin{pmatrix}
P/2 & P/2 &0&0\\
0&0&P/2&P/2
\end{pmatrix},\nonumber\\
C_3&=\begin{pmatrix}
P & 0 &0&0\\
0&0&P&0
\end{pmatrix}.
\end{align} 
Code $C_1$ corresponds to on/off keying strategy and its error probability as a function of $P$ is illustrated in red. Code $C_2$ is a symmetric strategy and its error probability is illustrated in green for both channel. The last coding scheme is $C_3$ that is the optimal one (according on Theorem~\ref{t1}) for both channels and its error probability as a function of $P$ is illustrated in blue. This figure shows that the on/off keying codeword has the worst performance, and its gap to the other two strategies is more significant when there is no channel memory.

We now turn to the case of $N\leq K$. Here, we only consider the special case of $N=K=2$.

\begin{theorem}
\label{t2}
Consider a discrete Poisson channel with memory  of order $K=2$ and dark noise $d>0$. Assume that the blocklength $N=2$ and we have a  total-power  constraint $ x_{10}+x_{11}\leq P$ and $x_{20}+x_{21}\leq P$. Then, there are optimal codewords of the form
\begin{align}
\begin{pmatrix}
\boldsymbol{x_1}\\
\boldsymbol{x_2}
\end{pmatrix}=\left( \begin{matrix}
P & 0 \\
 0& x
\end{matrix}
\right)\label{eqnLL1}
\end{align} 
for some $ 0 \leq x \leq P$.
\end{theorem}
Proof of the above theorem is given in Section~\ref{five2}.

\begin{figure*}[!t]
\centering
\subfloat[$P=10$]{\includegraphics[width=3.25 in]{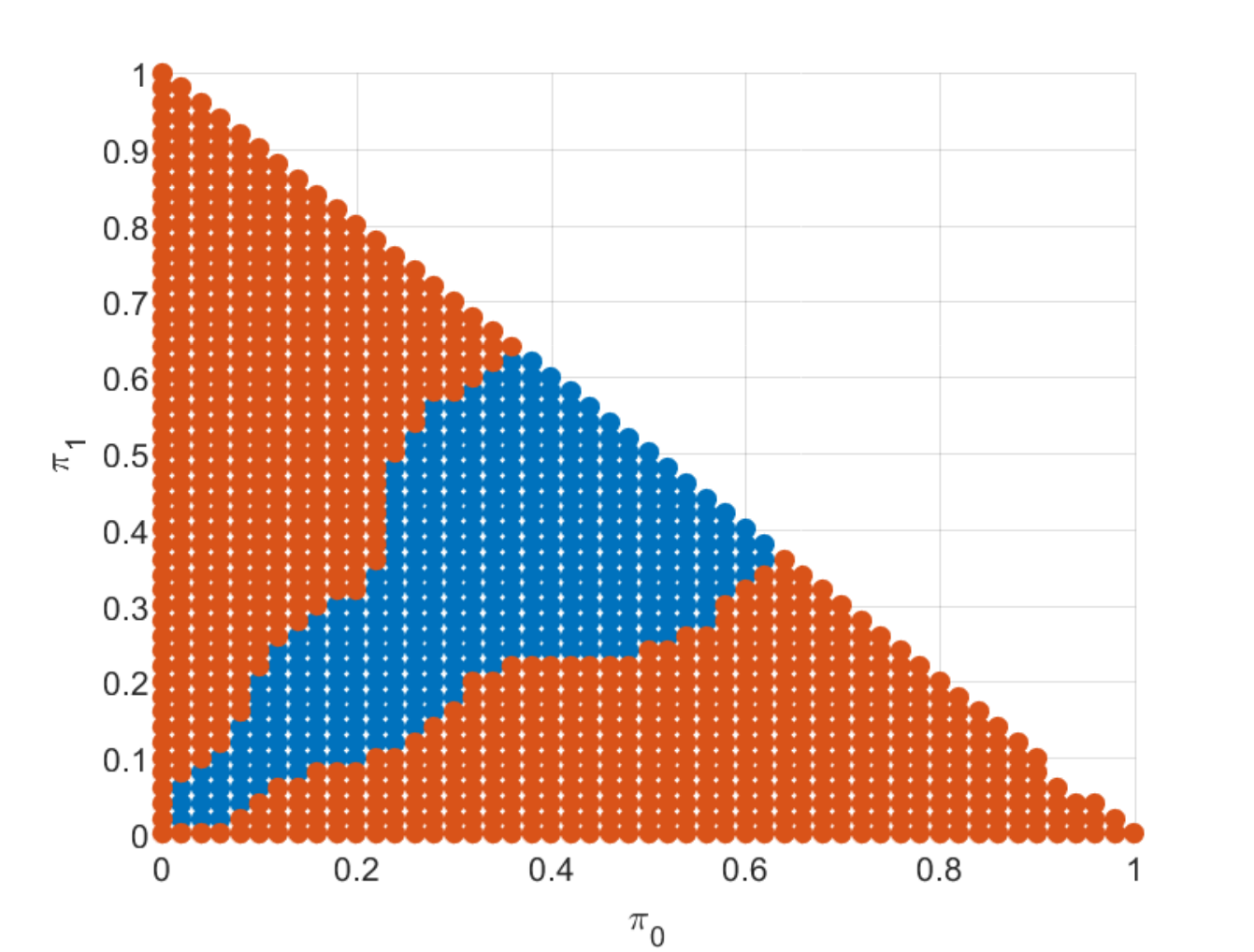}%
\label{abi1}}
\hfil
\subfloat[ $P=8$ ]{\includegraphics[width=3.25 in]{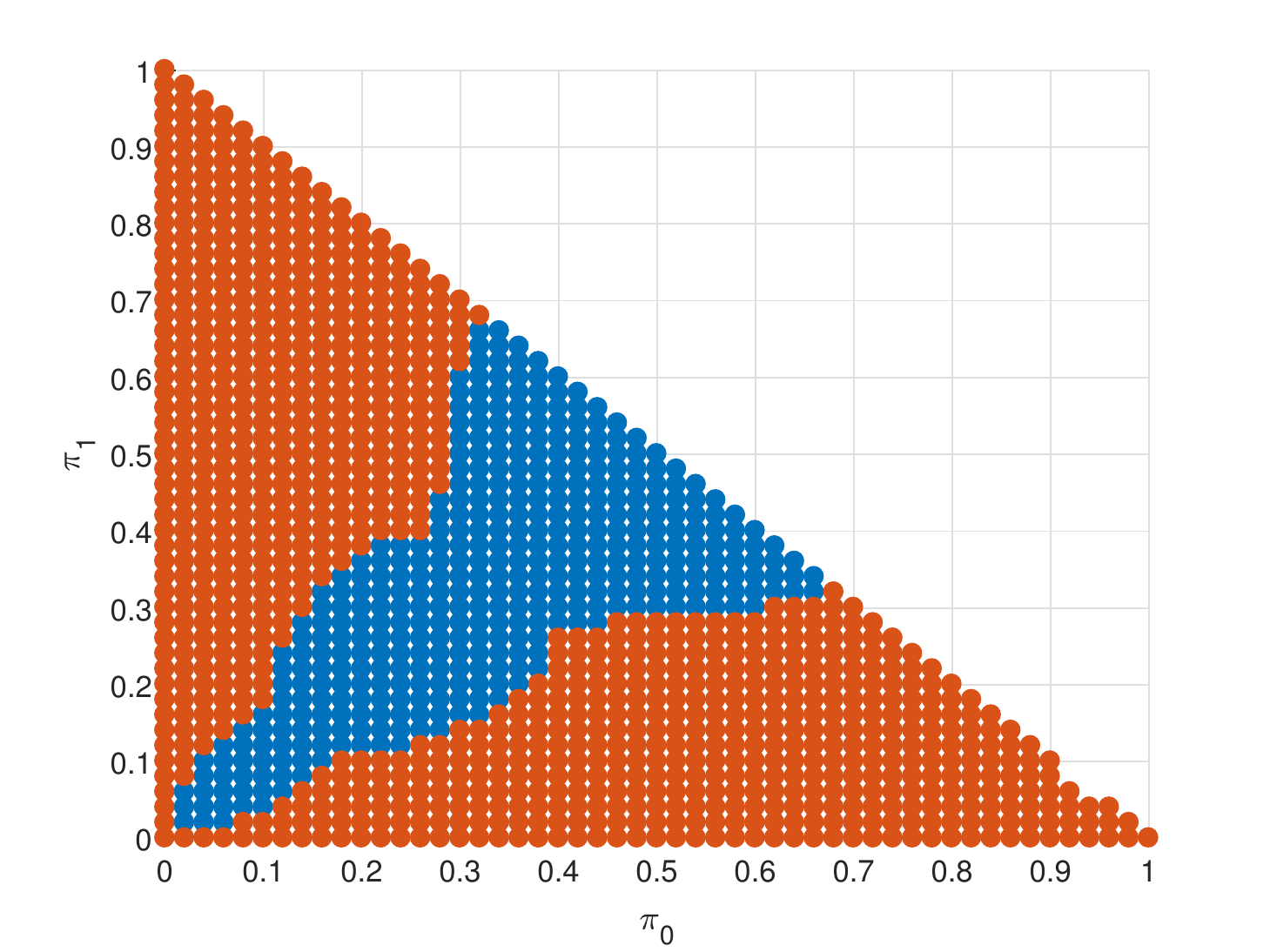}%
\label{abi2}}
\caption{ Comparison of the error probabilities of the codes stated in \eqref{eqabi1} and \eqref{eqabi2} for channel coefficients $0 \leq \pi_0$, $0 \leq \pi_1$ such that $\pi_0+\pi_1\leq 1$. A point $(\pi_0,\pi_1)$ is colored with blue if the error of the code \eqref{eqabi1} is less than the error of \eqref{eqabi2} and is colored with orange otherwise.  }
\label{abi}
\end{figure*}

\begin{figure}
\includegraphics[width=.45\textwidth]{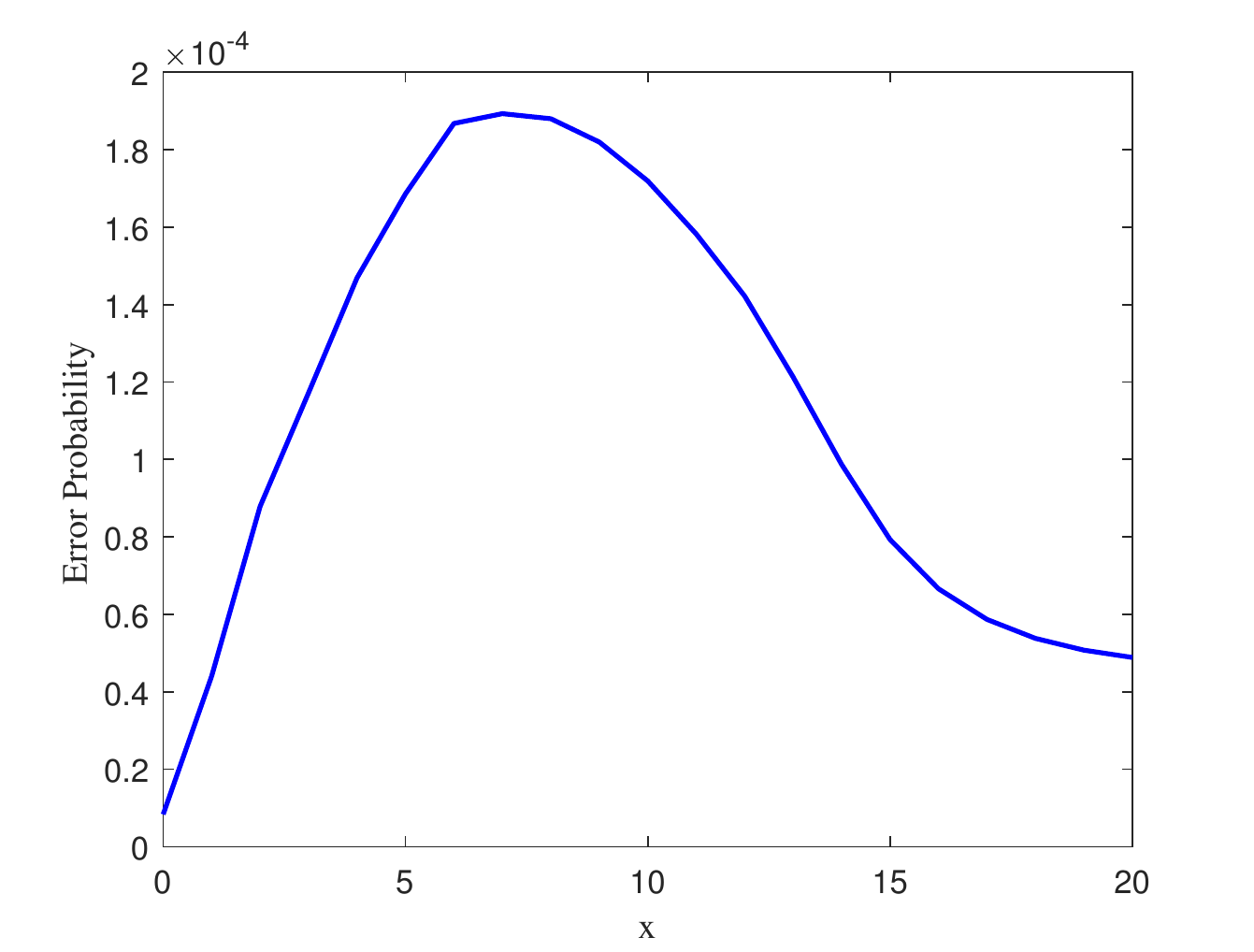}
\centering
\caption{The error probability of the code given in \eqref{eqnLL1} as a function of $x$ for $x\in[0,P]$ where $P=20$. In this example, we use $\boldsymbol{\pi}=\left(0.5, 0.5\right)$ and $d=0.1$.}
\label{fig11}
\end{figure}

\begin{example}\label{example6}
Consider the case of $N=K=2$, $P=10$ and $d=0.5$. Then, numerical simulation shows that for channel coefficients $\left[\pi_0,\pi_1\right]=\left[0.6, 0.4\right]$, the optimal codewords are 
\begin{align}
\begin{pmatrix}
\boldsymbol{x_1}\\
\boldsymbol{x_2}
\end{pmatrix}=\left( \begin{matrix}
P & 0 \\
 0& 0
\end{matrix}
\right),
\label{eqabi1}
\end{align}
with an error probability of $0.0092$. In this example, the codeword pair 
 \begin{align}
 \begin{pmatrix}
\boldsymbol{x_1}\\
\boldsymbol{x_2}
\end{pmatrix}=\left( \begin{matrix}
P & 0 \\
 0& P
\end{matrix}
\right),
\label{eqabi2}
\end{align}
has an error probability of 
$0.0095$ and is not optimal. This shows that in this example the following holds: (i) the codeword structure given in \eqref{eqn1} is not optimal , and (ii) as $x_{11}=x_{21}=0$, the error probability for $N=2$ is the same as the error probability for blocklength one. { Next, in Fig~\ref{abi}, for the points of the set $\{(\pi_0,\pi_1) \ \vert \ 0\leq\pi_0 \leq 1, \ 0\leq\pi_1 \leq 1, \ \pi_0+\pi_1\leq 1\},$ the error probability of the code in \eqref{eqabi1} is compared with the error probability of the code in \eqref{eqabi2}, for $P=10$ (Fig.~\ref{abi1}) and $P=8$ (Fig.~\ref{abi2}) with $d=0.5$ for the both cases. A point is colored with blue if the first error is less than the second one and is colored with orange otherwise. Note that the patterns are symmetrical with respect to line $\pi_0-\pi_1=0$. This is due to the particular structure of the codes in \eqref{eqabi1} and \eqref{eqabi2} whose error probabilities will be  symmetric functions of $\pi_0$ and $\pi_1$. 
Note that the code given in \eqref{eqabi2} is the same as the one given in Theorem \ref{t1}. However, the assumption of Theorem \ref{t1} ($N\geq K+1$) is violated here. This figure confirms that the code given in Theorem \ref{t1} may fail to be optimal without this assumption.
}
\end{example}
\begin{remark}\label{remarkN}
The proof of Theorem~\ref{t2} shows that any codeword that is ``locally" optimal (i.e., is not improved by local changes) must be of the following form:
\begin{align*}
\begin{pmatrix}
\boldsymbol{x_1}\\
\boldsymbol{x_2}
\end{pmatrix}=\left( \begin{matrix}
P & 0 \\
 0& x
\end{matrix}
\right)
\end{align*} 
 Figure~\ref{fig11} plots the error probability of the code given in Theorem~\ref{t2} in terms of $x$ for the following channel $\boldsymbol{\pi}=\left[0.5, 0.5\right]$, with dark noise $d=0.1$. The curve is not convex which indicates that the problem of finding the optimal codewords is not a convex optimization problem. 
\end{remark}

\subsection{Code design subject to a  peak-power  constraint}
In this section, we only consider a  peak-power  constraint on the individual codewords ($P=\infty$ but $A$ is finite). 
\begin{figure*}[!t]
\centering
\subfloat[ $K=2$, $\boldsymbol{\pi}=(1/2,1/2)$ and $d=0.25$.]{\includegraphics[width=3.25 in]{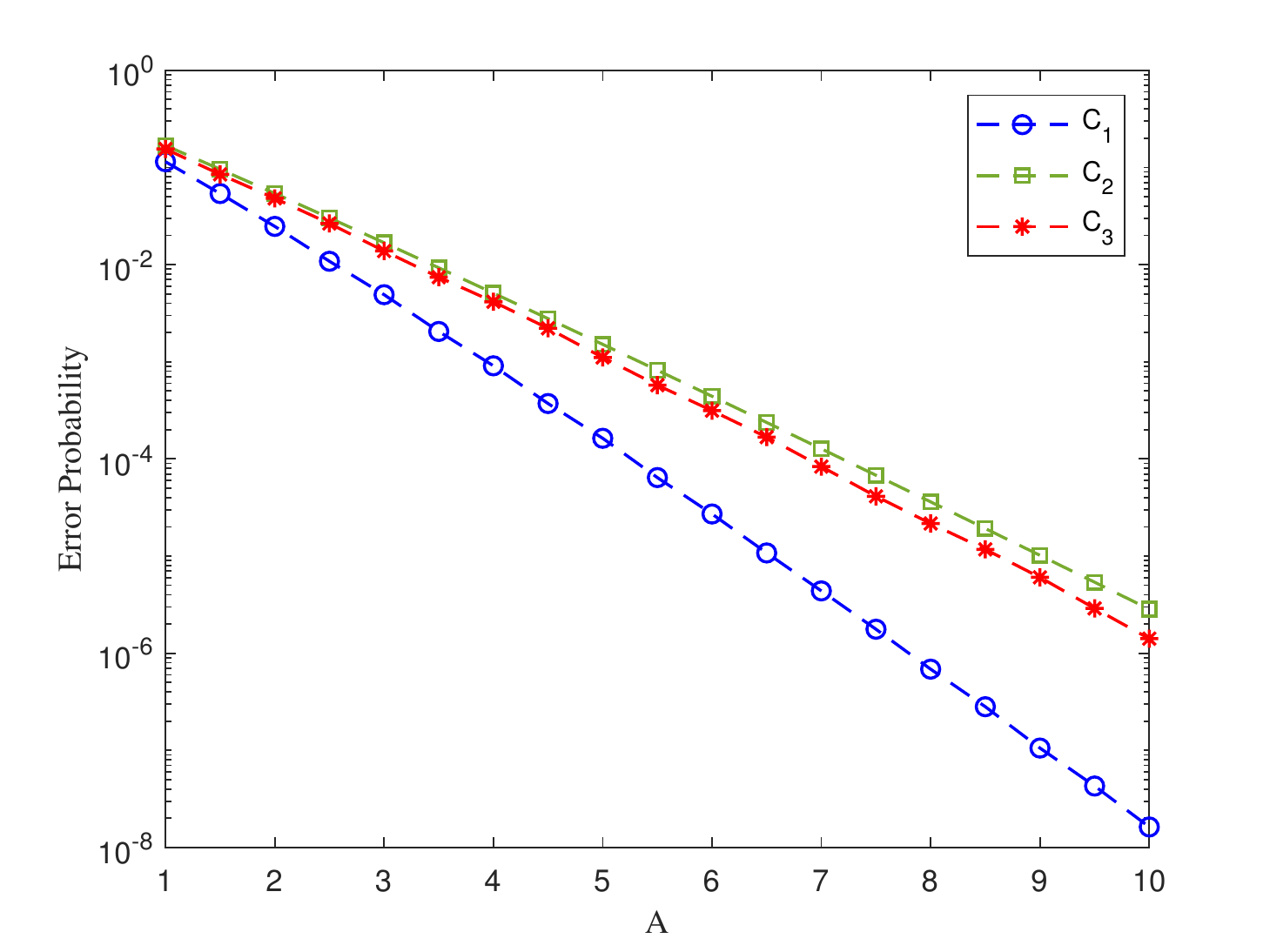}%
\label{figa1}}
\hfil
\subfloat[ $K=1$, $\boldsymbol{\pi}=(1)$ and $d=0.25$.]{\includegraphics[width=3.25 in]{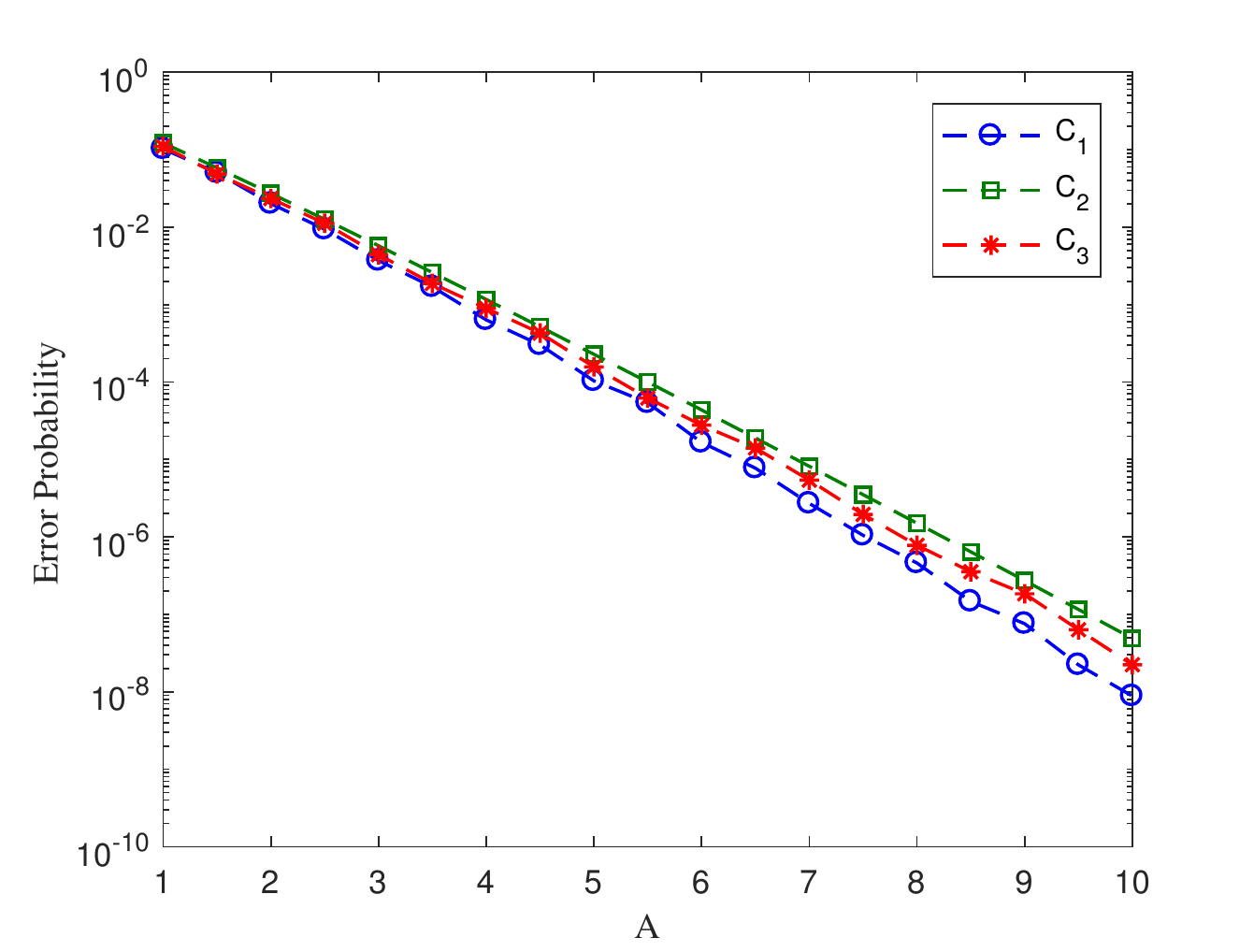}%
\label{figa2}}
\caption{Comparison of the error probability of the three coding schemes mentioned in \eqref{newa} for two different channels.}
\label{figa_sim}
\end{figure*}

\begin{theorem}
\label{t3}
Consider a discrete Poisson channel with memory  of order $K\geq 1$ and dark noise $d > 0$. Then 
there exists a constant $A^*$ depending on the blocklength $N$, the channel coefficients and dark noise level $d$ such that the following holds: for codewords with  peak-power  constraint $A \geq A^*$, the optimal codeword pair is
\begin{align*}
\begin{pmatrix}
\boldsymbol{x_1}\\
\boldsymbol{x_2}
\end{pmatrix}= \Bigg( \overbrace{ \begin{matrix}
A & A &  \cdots  & A  \\     
 0& 0& \cdots  &0                                          
\end{matrix}}^N  \Bigg) .
\end{align*}
More specifically, the above code construction is optimal if $A$ satisfies \eqref{c1-eq1} given at the top of
page~\pageref{c1-eq1}.
\newcounter{storeeqcounter}
\newcounter{tempeqcounter}
\setcounter{storeeqcounter}{\value{equation}}
\addtocounter{equation}{1}
\end{theorem}
Proof of the above theorem is given in Section~\ref{five3}.

In Fig.~\ref{figa_sim} we compare the error probabilities of three coding schemes. The channels used for this plot are the as those used in Fig.~\ref{figp_sim}.  The three codes are as follows:
\begin{align}
\label{newa}
C_1&=\begin{pmatrix}
A & A &A&A\\
0&0&0&0
\end{pmatrix}, \nonumber\\
C_2&=\begin{pmatrix}
A & A &0&0\\
0&0&A&A
\end{pmatrix},\nonumber\\
C_3&=\begin{pmatrix}
A & A &A&0\\
0&0&0&A
\end{pmatrix}.
\end{align} 
Code $C_1$ corresponds to an on/off keying strategy which is the optimal code (according to Theorem~\ref{t3}) for both channel models when $A$ is sufficiently large.
Its error probability is illustrated in blue. Note that  the code $C_1$ generally beats the codes $C_2$ and $C_3$ even for small values of $A$. 
However, in Fig.~\ref{zoom} we we zoom Fig.~\ref{figa1} around $A=1.5$. As one can observe, for $A=1.5$, the error probability of code $C_1$ is $0.0504$ whereas the error probability of the code $C_3$ is $0.0485$ which is smaller. This indicates that for small values of $A$, it is not necessarily true that the code $C_1$ has the smallest error probability among all possible codes.

\begin{figure}
 \includegraphics[width=.45\textwidth]{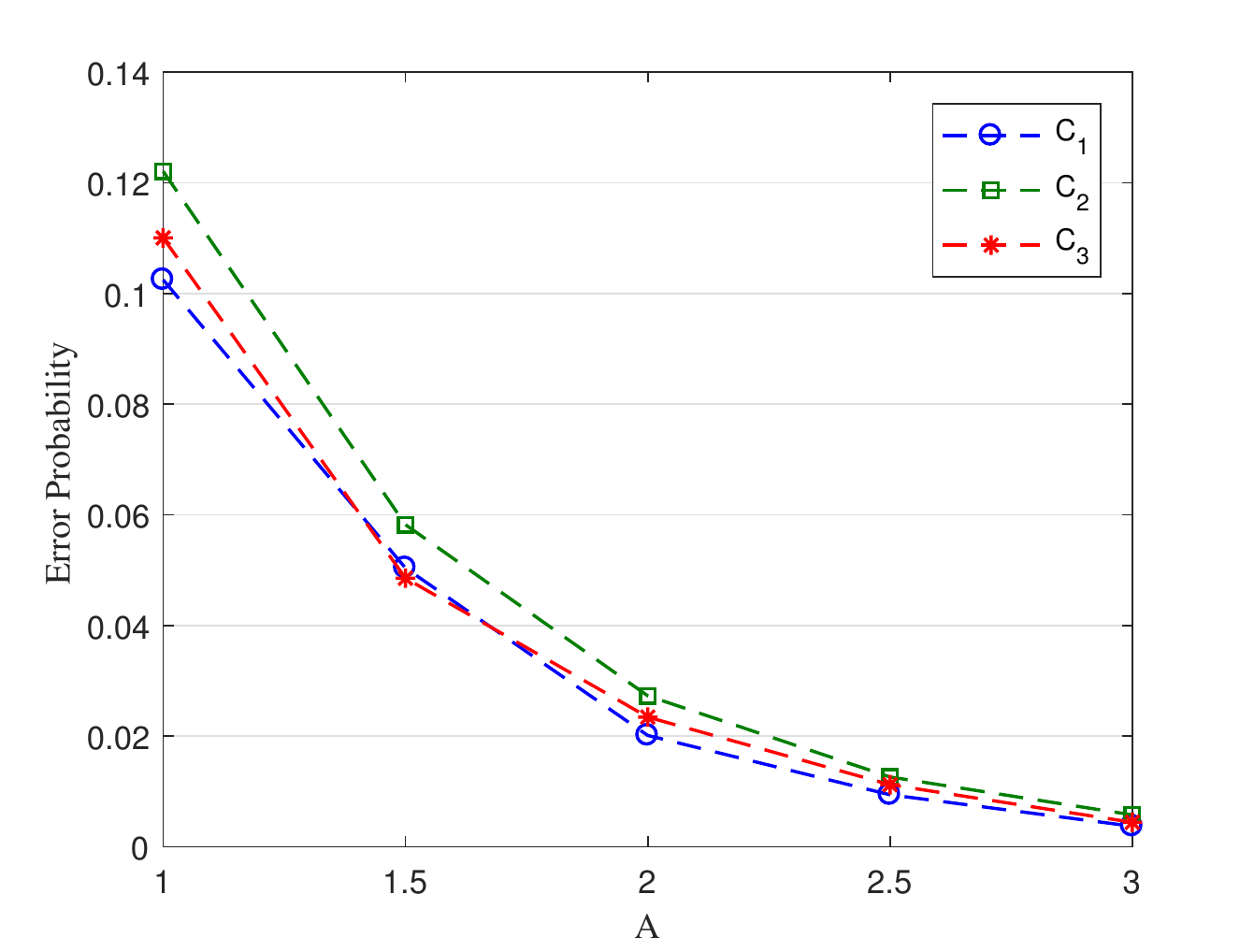}
 \centering
 \caption{The error probabilities of the three coding schemes mentioned in \eqref{newa}, for the channel with $K=1$, $\boldsymbol{\pi}=(1)$ and $d=0.25$. }
 \label{zoom}
 \end{figure}
\subsection{Code design subject to both peak and  total-power  constraints}
Assume that $K=1$ (no channel memory) and that we have both peak and  total-power  constraints and their ratio is $\beta$, \emph{i.e.,} $P=A\beta$. If
 $N \leq \lfloor \beta \rfloor $, the amplitude constraint implies the total-power constraint (the total power constraint equation becomes inactive), and a solution can be found. This case is similar to Theorem~\ref{t3} with $K=1$. The case of $N \geq \lfloor \beta \rfloor +1$ is treated in the following theorem in the  high-power  regime. For simplicity we assume $\pi_0=1$ in the following theorem (a similar result can be shown for arbitrary values of $\pi_0$).

\begin{figure*}[!t]
\normalsize
\setcounter{tempeqcounter}{\value{equation}}
\begin{align}
\setcounter{equation}{\value{storeeqcounter}}
\label{c1-eq1}
    &e^{-(N+K-1)\left(\frac {NA\left(\sum_{i=0}^{K-1} \pi_i\right)}{N+K-1}+d\right)}\nonumber\\
    &\quad\times \left( \small{\frac{e\left(\frac {NA\left(\sum_{i=0}^{K-1} \pi_i\right)}{N+K-1}+d\right)\log\left({\left(\frac {NA\left(\sum_{i=0}^{K-1} \pi_i\right)}{N+K-1}+d\right)}/{d}\right)}{\frac {NA\left(\sum_{i=0}^{K-1} \pi_i\right)}{N+K-1}}} \right)^{{(N+K-1)}/{\log\left({\left(\frac {NA\left(\sum_{i=0}^{K-1} \pi_i\right)}{N+K-1}+d\right)}/{d}\right)}}
 \nonumber\\
 &\leq \frac12 e^{-\left(\max_{j\in \{ 0, \ldots, K-1 \} }\left(-A\pi_j +NA\left(\sum_{i=0}^{K-1} \pi_i\right)\right)+(N+K-1)d\right)}.
\end{align}
\setcounter{equation}{\value{tempeqcounter}} 
\hrulefill
\vspace*{4pt}
\end{figure*}

\begin{theorem}
\label{t4}
Consider a discrete Poisson channel without memory  (i.e., $K= 1$), $\pi_0=1$ and dark noise $d > 0$. Fix some constant $\beta>0$ and { blocklength $N \geq \lfloor \beta \rfloor+1$}. Then there is another constant $A^*$ which depends only on $\beta$, $d$ and {$N$} such that the following holds: for a given peak-power constraint $A\geq A^*$ and  total-power  constraint $P=A\beta$, the optimal codeword pair $\boldsymbol{x_1}$ and $\boldsymbol{x_2}$ is 
\begin{align*}
{x_{1i}}&=\kappa_i, \qquad 0\leq i\leq N-1
\\
{x_{2i}}&=\eta_i, \qquad 0\leq i\leq N-1,
\end{align*}
where $\boldsymbol\kappa$ and $\boldsymbol\eta$ are sequences of infinite length defined as \eqref{c1-th7} given at the top of
page~\pageref{c1-th7}, for noninteger $\beta$, where $\lfloor\beta\rfloor$ and $\{\beta\}$ are the integer and fractional parts of $\beta$ respectively. For integer $\beta$, the sequences $\boldsymbol\kappa$ and $\boldsymbol\eta$ are defined as \eqref{c1-th71} given at the top of
page~\pageref{c1-th71}. More specifically, the code construction is optimal if $A$ satisfies
\newcounter{storeeqcounterth7}
\newcounter{tempeqcounterth7}
\setcounter{storeeqcounterth7}{\value{equation}}
\addtocounter{equation}{1}
\newcounter{storeeqcounterth71}
\newcounter{tempeqcounterth71}
\setcounter{storeeqcounterth71}{\value{equation}}
\addtocounter{equation}{1}
\begin{figure*}[!t]
\normalsize
\setcounter{tempeqcounterth7}{\value{equation}}
\begin{align}
\setcounter{equation}{\value{storeeqcounterth7}}
\label{c1-th7}
    \begin{pmatrix}
\boldsymbol{\kappa}\\
\boldsymbol{\eta}
\end{pmatrix}=\color{black} \Bigg( \color{black} \overbrace{\begin{matrix}
A & A & \cdots & A & A\{\beta\}  \\    
 0& 0 &\cdots & 0 & 0                         
\end{matrix}}^{\lfloor \beta \rfloor +1}
\begin{matrix}
&\\
&
\end{matrix}
\overbrace{\begin{matrix}
  0 & 0  &\cdots &0& 0\\ 
  A & A &\cdots &A& A\{\beta\}
\end{matrix}}^{\lfloor \beta \rfloor +1}
\begin{matrix}
&\\
&
\end{matrix}
\begin{matrix}
0 & 0 \cdots \\
0 & 0 \cdots 
\end{matrix}
\Bigg) 
\end{align}
\setcounter{equation}{\value{tempeqcounterth7}} 
\hrulefill
\vspace*{4pt}
\end{figure*}
\begin{figure*}[!t]
\normalsize
\setcounter{tempeqcounterth71}{\value{equation}}
\begin{align}
\setcounter{equation}{\value{storeeqcounterth71}}
\label{c1-th71}
    \begin{pmatrix}
\boldsymbol{\kappa}\\
\boldsymbol{\eta}
\end{pmatrix}=\color{black} \Bigg( \color{black} \overbrace{\begin{matrix}
A & A & \cdots & A & A  \\    
 0& 0 &\cdots & 0 & 0                         
\end{matrix}}^{\beta}
\begin{matrix}
&\\
&
\end{matrix}
\overbrace{\begin{matrix}
  0 & 0  &\cdots &0& 0\\ 
  A & A &\cdots &A& A
\end{matrix}}^{ \beta }
\begin{matrix}
&\\
&
\end{matrix}
\begin{matrix}
 0 & 0 \cdots \\
 0 & 0 \cdots 
\end{matrix}
 \Bigg) 
\end{align}
\setcounter{equation}{\value{tempeqcounterth71}} 
\hrulefill
\vspace*{4pt}
\end{figure*}
\begin{align}
    \frac12 e&^{-\left(A(\beta-\Gamma)+Nd\right)}\nonumber\\
    \geq & \ e^{-n\left(\frac {A\beta}{n}+d\right)}\nonumber\\
       &\times \left( \small{\frac{e\left(\frac {A\beta}{n}+d\right)\log\left({\left(\frac {A\beta}{n}+d\right)}/{d}\right)}{\frac {A\beta}{n}}} \right)^{{n}/{\log\left({\left(\frac {A\beta}{n}+d\right)}/{d}\right)}}\label{eqnNNNNNNNN22}
\end{align}
where
\begin{align}
  \Gamma\triangleq \begin{cases} \{\beta\} &\mbox{if } \{\beta\} > 0 \\
1 & \mbox{if } \{\beta\} = 0, \end{cases}\label{eqnNAAAAA}
\end{align}
and
\begin{align}
  n\triangleq \begin{cases} \lfloor \beta \rfloor +1 &\mbox{if } \{\beta\} > 0 \\
\beta & \mbox{if } \{\beta\} = 0. \end{cases}
\end{align}

\end{theorem}
Proof of the above theorem is given in Section~\ref{five4}.

\section{Proofs}
\label{five}
In the proofs of Theorem~\ref{t1} and Theorem~\ref{t3}, we use the majorization and some related concepts  that are reviewed in Subsection~\ref{two}.

\subsection{Majorization}
\label{two}

Majorization is a preorder on vectors of real numbers.

We say that a sequence 
$\boldsymbol{a} \in \mathbb{R}^n$ weakly majorizes $\boldsymbol{b} \in \mathbb{R}^n$
from below and write it as $ \boldsymbol{a} \succ_{\w} \boldsymbol{b}$ (equivalently, we say that $\boldsymbol{b}$ is weakly majorized by $\boldsymbol{a}$ from below, written as $ \boldsymbol{b} \prec_{\w} \boldsymbol{a}$)   if and only if 
\begin{align*}
\sum_{i=1}^k {a}_i^{\downarrow} \geq \sum_{i=1}^k {b}_i^{\downarrow} \quad k=1, \ldots, n, 
\end{align*}
where $\boldsymbol{{a}}^{\downarrow}$ ($\boldsymbol{{b}}^{\downarrow}$) is
the vector with the same components with $\boldsymbol{a}$ ($\boldsymbol{b}$) that is sorted in the descending order.
If $\boldsymbol{a} \succ _{\w}\boldsymbol{b}$ and in addition $ \sum _{i=1}^{n}a_{i}=\sum _{i=1}^{n}b_{i}$ we say that $ \mathbf {a} $ majorizes $ \mathbf {b}$ written as $ \mathbf {a} \succ \mathbf {b} $.
 
\begin{definition}
A Robin Hood operation on a nonnegative sequence $\mathbf{a}$ replaces two elements $a_i$ and $a_j <a_i$ by $a_i - \epsilon$ and $a_j + \epsilon$, respectively, for some $\epsilon \in \left(0, a_i - a_j\right)$ \cite{arnold}. An Anti-Robin Hood operation (AR) does the opposite on nonnegative sequences. That is, for a sequence $\mathbf{a}$ with nonnegative entries,  it replaces $ a_i$ and $a_j <a_i$ by $a_i + \epsilon$ and $a_j - \epsilon$, for some $\epsilon \in \left(0, a_j\right]$.  
\end{definition}
\begin{lemma}[\cite{arnold}(p. 11)]
\label{l5}
Take two nonnegative sequences $\boldsymbol{\lambda}$ and $\boldsymbol{\pi}$ of the same length satisfying $\boldsymbol{\lambda} \prec \boldsymbol{\pi}$. Then starting from $\boldsymbol{\lambda}$, we can produce $\boldsymbol{\pi}$ by a finite number of AR operations. Equivalently we can reach $\boldsymbol{\lambda}$ from $\boldsymbol{\pi}$ with the finite sequence of Robin Hood operations.
\end{lemma}

\subsection{Proof of Theorem~\ref{t1}}
\label{five1}

The error probability depends on the two codewords $\bxo$ and $\bxoo$ through vectors $\blambda=\boldsymbol{x_1}* \boldsymbol{\pi}$ and $\bmu=\boldsymbol{x_2}* \boldsymbol{\pi}$ as defined in equations \eqref{eqnA2} and \eqref{eqnA3}. In particular, for a given dark noise $d$, we can view the error probability as a function of the pair $\left(\boldsymbol{\lambda},\boldsymbol{\mu}\right)$. We can write this as $\mathbb{P}_{\e}\left(\boldsymbol{\lambda},\boldsymbol{\mu}\right)$.

Let $\mathcal{T}$ be the set of pairs $\left(\boldsymbol{\lambda},\boldsymbol{\mu}\right)$ for which one can find nonnegative  sequences $\boldsymbol{x_1}, \boldsymbol{x_2}$ satisfying $\boldsymbol{\lambda}=\boldsymbol{x_1}* \boldsymbol{\pi}$ and $\boldsymbol{\mu}=\boldsymbol{x_2}* \boldsymbol{\pi}$ and 
the total-power constraint $P$ in \eqref{jadid1},
\begin{align*}
\mathcal{T}
\triangleq \big\{ \left(\boldsymbol{\lambda},\boldsymbol{\mu}\right) \in (\mathbb{R}_0^+ \times \mathbb{R}_0^+)^{N+K-1}: \exists \ \boldsymbol{x}_1,\boldsymbol{x}_2 \in  (\mathbb{R}_0^+)^N\\
 \text{satisfying \eqref{jadid1} and \eqref{eqnA2}, \eqref{eqnA3} } \big\}.
\end{align*}
Then, the problem can be expressed as follows:
\begin{align}
\label{op1}
\argmin_{\left(\boldsymbol{\lambda},\boldsymbol{\mu}\right)\in\mathcal{T}} \ \mathbb{P}_{\e}\left(\boldsymbol{\lambda},\boldsymbol{\mu}\right).
\end{align}
To solve this problem, we first relax the above optimization problem by defining a set $\mathcal{T}'$ satisfying
$\mathcal{T}\subset\mathcal{T'}$, and consider
\begin{align}
\label{op13}
\argmin_{\left(\boldsymbol{\lambda},\boldsymbol{\mu}\right)\in\mathcal{T'}} \ \mathbb{P}_{\e}\left(\boldsymbol{\lambda},\boldsymbol{\mu}\right).
\end{align}
We find a solution 
$\left(  \boldsymbol{\lambda}^\star,\boldsymbol{\mu}^\star \right) $ of the relaxed problem in \eqref{op13} and verify that the optimizer pair
$\left(\boldsymbol{\lambda}^\star,\boldsymbol{\mu}^\star\right)$ belongs to $\mathcal{T}$. This shows that the answers to the original and relaxed problems in \eqref{op1} and \eqref{op13} are the same. We highlight a crucial difference between the optimization problems in  \eqref{op1} and \eqref{op13} that is very helpful to our proof: given a pair $\left(\boldsymbol{\lambda},\boldsymbol{\mu}\right)\in\mathcal{T'}$, if we apply the same permutation on the sequences $\boldsymbol{\lambda}$ and $\boldsymbol{\mu}$, we obtain a pair in $\mathcal{T'}$. However, given a pair $\left(\boldsymbol{\lambda},\boldsymbol{\mu}\right)\in\mathcal{T}$, it is not necessarily true that we remain in $\mathcal{T}$ after permuting the two sequences. 

Let $$\mathcal{T}'
\triangleq\big\{\left(\boldsymbol{\lambda},\boldsymbol{\mu}\right) \in (\mathbb{R}_0^+ \times \mathbb{R}_0^+)^{N+K-1}: \boldsymbol{\lambda},\boldsymbol{\mu} \prec_\text{w}  \left[P\boldsymbol{\pi},\boldsymbol{0}_{N-1} \right]\big\}.$$
That is, $\mathcal{T}'$ is the set of pairs $\left(\boldsymbol{\lambda},\boldsymbol{\mu}\right)$ of sequences of length $N+K-1$ that are weakly majorized from below by the sequence $ \left[P\boldsymbol{\pi},\boldsymbol{0}_{N-1} \right]$. Here,  $ \left[P\boldsymbol{\pi},\boldsymbol{0}_{N-1} \right]$ is a sequence of length $N+K-1$ formed by concatenating the sequence $P\boldsymbol{\pi}$ of length $K$ with the all zero sequence $\boldsymbol{0}_{N-1}$ of length $N-1$. 

\begin{lemma} The relation $\mathcal{T}\subset\mathcal{T'}$ holds.\label{l4}
\end{lemma}
\begin{proof}  We need to show that if $
\boldsymbol{\lambda}=\boldsymbol{x_1}* \boldsymbol{\pi}$ and $
\boldsymbol{\mu}=\boldsymbol{x_2}* \boldsymbol{\pi}
$ for some nonnegative  sequences $\boldsymbol{x_1}$ and $\boldsymbol{x_2}$ satisfying  
 the total-power constraint $P$ in \eqref{jadid1}  then
$$ \boldsymbol{\lambda},\boldsymbol{\mu} \prec_{\w}  \left[P\boldsymbol{\pi},\boldsymbol{0}_{N-1} \right].$$

Let $$\bar{x}=\sum_{i=0}^{N-1} x_{1i}\leq P.$$  We define a function $Z(\cdot, \cdot)$ that takes in an integer and a sequence of real numbers, and outputs another sequence of real numbers as follows: for $0\leq i\leq N-1$, let
$$Z\left(\bar{x}\boldsymbol{\pi},i\right) \triangleq \left[\boldsymbol{0}_i, \bar{x}\boldsymbol{\pi}, \boldsymbol{0}_{N-i-1}\right]$$
be a sequence of length $N+K-1$ formed by padding zeros to the beginning and end of the sequence $\bar{x}\boldsymbol{\pi}$. Then,\begin{align*}
\boldsymbol{\lambda}=\boldsymbol{x_1}* \boldsymbol{\pi}= \sum_{i=0}^{N-1} \frac{x_{1i}}{\bar{x}} Z\left(\bar{x}\boldsymbol{\pi},i\right).
\end{align*}
Observe that $\boldsymbol{\lambda}$ is expressed as a convex combination of the sequences $Z\left(\bar x\boldsymbol{\pi},i\right)$. Each of the sequences $Z\left(\bar{x}\boldsymbol{\pi},i\right)$ for $0\leq i\leq N-1$ is a permutation of $\left[\bar{x}\boldsymbol{\pi},\boldsymbol{0}_{N-1}\right]$. Therefore, from \cite[Theorem~2.1]{arnold} we conclude that $ \boldsymbol{\lambda} \prec_{\w} \left[\bar{x}\boldsymbol{\pi},\boldsymbol{0}_{N-1}\right]$. Since  $\left[\bar{x}\boldsymbol{\pi},\boldsymbol{0}_{N-1}\right]\prec_{\w} \left[P\boldsymbol{\pi},\boldsymbol{0}_{N-1}\right]$, we obtain 
$ \boldsymbol{\lambda}\prec_{\w} \left[P\boldsymbol{\pi},\boldsymbol{0}_{N-1}\right].$ The proof for 
$\boldsymbol{\mu} \prec_{\w} \left[P\boldsymbol{\pi},\boldsymbol{0}_{N-1}\right]$ is similar. 
\end{proof}

We claim that an optimal solution for \eqref{op13} is as \eqref{isiopt} given at the top of
page~\pageref{isiopt}. This claim would conclude the proof of the theorem since $\left(\boldsymbol{\lambda}^\star,\boldsymbol{\mu}^\star\right)\in\mathcal{T}$ as $\boldsymbol{\lambda}^\star=\boldsymbol{x_1}* \boldsymbol{\pi}$ and $\boldsymbol{\mu}^\star=\boldsymbol{x_2}* \boldsymbol{\pi}$ for  $\left(\boldsymbol{x_1},\boldsymbol{x_2}\right)$ given in the statement of the theorem.
\newcounter{storeeqcountern1}
\newcounter{tempeqcountern1}
\setcounter{storeeqcountern1}{\value{equation}}
\addtocounter{equation}{1}
\begin{figure*}[!t]
\normalsize
\setcounter{tempeqcountern1}{\value{equation}}
\begin{align}
\setcounter{equation}{\value{storeeqcountern1}}
\label{isiopt}
\begin{pmatrix}
\boldsymbol{\lambda^\star}\\
\boldsymbol{\mu^\star}
\end{pmatrix}=\left( \begin{matrix}
P\pi_0  & \color{black} \cdots \color{black} & P\pi_{K-2} & P\pi_{K-1}& 0& 0&  \color{black} \cdots \color{black} &0 & \color{black} \cdots \color{black} &0\\
 0& 0&\color{black} \cdots \color{black} &0 & P\pi_0  &\color{black} \cdots \color{black} &P\pi_{K-2} & P\pi_{K-1} & \color{black} \cdots \color{black} &0
\end{matrix}
\right)
\end{align}
\setcounter{equation}{\value{tempeqcountern1}} 
\hrulefill
\vspace*{4pt}
\end{figure*}

It remains to prove that $\left(\boldsymbol{\lambda^\star}, 
\boldsymbol{\mu^\star}\right)$ is an optimal solution. To show this, we start from an arbitrary pair $\left(\boldsymbol{\lambda},\boldsymbol{\mu}\right)$ and alter the sequences $\boldsymbol{\lambda}$ and $\boldsymbol{\mu}$ in a finite number of steps such that 
\begin{itemize}
\item the error probability does not increase in each step;
\item we end up with the sequences $\left(\boldsymbol{\lambda^\star}, 
\boldsymbol{\mu^\star}\right)$.
\end{itemize}

Therefore, $\left(\boldsymbol{\lambda^\star}, \boldsymbol{\mu^\star}\right)$ is a solution with the minimum error probability. To see this, take  an arbitrary pair $(\boldsymbol{\hat\lambda}, \boldsymbol{\hat\mu})$ with the error probability $\mathbb{P}_{\e} (\boldsymbol{\hat\lambda},\boldsymbol{\hat\mu})$. Then we can alter the pair $(\boldsymbol{\hat\lambda}, \boldsymbol{\hat\mu})$ in a finite steps and end up with the pair $\left(\boldsymbol{\lambda^\star}, \boldsymbol{\mu^\star}\right)$ while the error probability does not increase in each step. As a result the error probability at the final step should be less than or equal to the error probability of the initial pair. In other words, $\mathbb{P}_{\e} (\boldsymbol{\lambda^\star},\boldsymbol{\mu^\star}) \leq \mathbb{P}_{\e} (\boldsymbol{\hat\lambda},\boldsymbol{\hat\mu}) $. Since the pair $(\boldsymbol{\hat\lambda}, \boldsymbol{\hat\mu})$ was arbitrary, we conclude that $\left(\boldsymbol{\lambda^\star}, \boldsymbol{\mu^\star} \right)$ is a minimizer of the error probability. Please note that this does not imply that the pair $\left(\boldsymbol{\lambda^\star}, \boldsymbol{\mu^\star} \right)$ is a \emph{unique} minimizer, and we do not make this claim in the statement of the theorem.

Take an arbitrary pair $\left(\boldsymbol{\lambda},\boldsymbol{\mu}\right)$ of nonnegative real vectors of length $N+K-1$. 
Let 
$\mathcal{A}=\{0\leq i\leq N+K-2: \lambda_i\geq \mu_i\}$, $\mathcal{A}^c=\{0,1,2,\ldots, N+K-2\} \backslash  \mathcal A$.
Without loss of generality we may assume that $\vert \mathcal{A} \vert \geq \left(N+K-1\right)/2 \geq K$, otherwise we can swap $\boldsymbol{\lambda}$ and $\boldsymbol{\mu}$ and follow the same argument. Furthermore, by applying a permutation on indices, without loss of generality we may assume that
$\mathcal{A}=\{0,1, \ldots , |\mathcal{A}|-1\}$ and 
$\mathcal{A}^c=\{|\mathcal{A}|,  \ldots , N+K-2\}$. We move from $\left(\boldsymbol{\lambda},\boldsymbol{\mu}\right)$ to $\left(\boldsymbol{\lambda}^\star,\boldsymbol{\mu}^\star\right)$ in three phases of steps:

\textbf{Phase 1:}
Using Lemma~\ref{l1} in Appendix~\ref{appendixa}, we can decrease $\lambda_i, i\in \mathcal{A}^c, $
and $\mu_i, i\in \mathcal{A}, $ to zero and produce $ \boldsymbol{\lambda'},\boldsymbol{\mu'}$ such that the error probability does not increase. Since 
$ \boldsymbol{\lambda'} \prec_{\w} \boldsymbol{\lambda}$ and $ \boldsymbol{\mu'} \prec_{\w} \boldsymbol{\mu}$, we obtain that
$ \boldsymbol{\lambda'},\boldsymbol{\mu'} \prec_{\w} \left[P\boldsymbol{\pi},\boldsymbol{0}_{N-1}\right]$.

 \textbf{Phase 2:}
Utilizing Lemma~\ref{l6}, $\boldsymbol{\lambda'}$ and $\boldsymbol{\mu'}$ of Phase 1 are majorized by \eqref{eq-n2} given at the top of
page~\pageref{eq-n2}, where $t,t' \in \{0,1,\ldots,K-2 \} $, $r_1 \leq P\pi_{t+1}$ and $r_2 \leq P\pi_{t'+1}$. We can move from $\boldsymbol{\lambda'}$ and $\boldsymbol{\mu'}$ 
to $\boldsymbol{\lambda''}$ and $\boldsymbol{\mu''}$ with AR operations using Lemma~\ref{l5}. These operations can be done while keeping  coordinates of the $\boldsymbol \lambda$ sequence zero on indices in $\mathcal{A}^c $ and $\boldsymbol \mu$ coordinates zero on $\mathcal{A} $. For this reason, Corollary~\ref{corrl1} guarantees that the error probability does not increase during these AR operations.
\newcounter{storeeqcountern2}
\newcounter{tempeqcountern2}
\setcounter{storeeqcountern2}{\value{equation}}
\addtocounter{equation}{1}
\begin{figure*}[!t]
\normalsize
\setcounter{tempeqcountern2}{\value{equation}}
\begin{align}
\setcounter{equation}{\value{storeeqcountern2}}
\label{eq-n2}
\left( \begin{matrix}
\boldsymbol{\lambda''}\\
\boldsymbol{\mu''}
\end{matrix} \right) =
\left( \begin{matrix}
P\pi_0  & \color{black} \cdots \color{black} & P\pi_{t} & r_1 &0 & \color{black} \cdots \color{black} & 0 & 0& 0&  \color{black} \cdots \color{black} \\
 0& 0&\color{black} \cdots \color{black} &0 &0&\color{black} \cdots \color{black} & 0 & P \pi_0  &\color{black} \cdots \color{black} &P\pi_{t'}
\end{matrix}
\begin{matrix}
&0 & \color{black} \cdots \color{black} &0 \\
& r_2 & \color{black} \cdots \color{black} &0
\end{matrix} \right)
\end{align}
\setcounter{equation}{\value{tempeqcountern2}} 
\hrulefill
\vspace*{4pt}
\end{figure*}

\textbf{Phase 3:}
Since the length of $\boldsymbol{\lambda''}$ and $\boldsymbol{\mu''}$ are greater than or equal to $2K$, based on Lemma~\ref{l1} we can increase the elements of $\boldsymbol{\lambda''}$ and reach $\left[\boldsymbol{\pi}_K,\boldsymbol{0}_{N-1} \right]$. Similarly (after reordering of indices) from $\boldsymbol{\mu''}$ we can reach $\left[\boldsymbol{0}_{N-1},\boldsymbol{\pi}_K\right]$ such that the error probability does not increase. With reordering the columns we reach
$\boldsymbol{\lambda^\star}, 
\boldsymbol{\mu^\star}$
 given in \eqref{isiopt}. 
The proof is complete. $\hfill  \square $

\subsection{Proof of Theorem~\ref{t2}}
\label{five2}

We assume that $\pi_0>0, \pi_1>0$ (otherwise the statement is immediate). 
The length of memory   is 2, hence $\boldsymbol{\pi}=\left[\pi_0,\pi_1\right]$ and the sequences $\boldsymbol{\lambda}=\boldsymbol{x_1}*\boldsymbol{\pi}$ and $\boldsymbol{\mu}=\boldsymbol{x_2}*\boldsymbol{\pi}$ are respectively as follows:
\begin{align}
\begin{matrix}
\boldsymbol{\lambda}=\big[x_{10} \pi_0,  & x_{11} \pi_0+x_{10} \pi_1, & x_{11} \pi_1\big],\\
\boldsymbol{\mu}=\big[x_{20} \pi_0,  & x_{21} \pi_0+x_{20} \pi_1, & x_{21} \pi_1\big].
\end{matrix}\label{lmu-num}
\end{align}
We utilize the fact that if a codeword pair $\left(\boldsymbol{x_1}, \boldsymbol{x_2}\right)$ is optimal, it must satisfy the necessary conditions given in Lemma~\ref{l7}. Assume that we have a pair of codewords that is \emph{not} of the following form:
\begin{align*}
\begin{pmatrix}
\boldsymbol{x_1}\\
\boldsymbol{x_2}
\end{pmatrix}=\left( \begin{matrix}
P & 0 \\
 0& x
\end{matrix}
\right).
\end{align*}

Then, one of the following three cases hold:
\begin{itemize}
 
 \item[1-] \emph{One of the codewords (without loss  of generality $\boldsymbol{x_1}$) satisfies the following conditions: $${x_{10}+ x_{11} <P, \ x_{10},x_{11}>0.}$$} Since $x_{10},x_{11}>0$, from Lemma~\ref{l7}
we have
\begin{align*}
  \begin{matrix}
D_0 \pi_0 +D_1 \pi_1=0, \\
D_1 \pi_0 +D_2 \pi_1=0.
\end{matrix}
\end{align*}
Hence $D_0=a$, $D_1=-a\frac{ \pi_0}{\pi_1}$ and 
$
D_2=a\big(\frac{\pi_0}{\pi_1}\big)^2$ 
for some $a \in \mathbb{R}$. Therefore,  the signs of $D_i, \  0 \leq i \leq 2$ have three different possibilities as follows:
\begin{align*}
\begin{matrix}
D_0>0, & D_1<0, & D_2>0; \\
D_0<0, & D_1>0, & D_2<0; \\
D_0=0, & D_1=0, & D_2=0 .
\end{matrix}
\end{align*}
For each case, we can use 
Lemma~\ref{l7} to find an equivalent condition for decision rule coefficients $a_i$. Using the definition of $a_i$ given in \eqref{eqn:LMu} and the definitions of $\boldsymbol{\lambda}$ and $\boldsymbol{\mu}$ given in \eqref{lmu-num}, we obtain  the following conditions for the three cases respectively: 
\begin{enumerate}
    \item $x_{10} <x_{20},    x_{11} < x_{21}$ and $x_{11} \pi_0 + x_{10} \pi_1 > x_{20} \pi_0 + x_{21} \pi_1;$
    \item $x_{10} >x_{20},   x_{11} > x_{21}$ and $ x_{11} \pi_0 + x_{10} \pi_1 < x_{20} \pi_0 + x_{21} \pi_1;$
    \item $
    x_{10}=x_{20}$ and $  x_{11}=x_{21}$.
\end{enumerate}
One can directly verify that in the first two cases, the inequalities contradict each other. The third case implies that the two codewords are equal,  which is clearly the worst possible choice for the two codewords. 

\item[2-]\emph{One of the codewords (without loss  of generality $\boldsymbol{x_1}$) satisfies the following conditions: $${x_{10}+ x_{11} =P, \ x_{10},x_{11}>0}.$$} 
From Lemma~\ref{l7} we have
\begin{align*}
 D_0 \pi_0 +D_1 \pi_1 =D_1 \pi_0 + D_2 \pi_1\leq 0,
\end{align*}
which results in
\begin{align*}
 \ D_2>D_1 \iff D_0 >D_1.
\end{align*}
From $\pi_1,\pi_2>0$, the sign of $D_i$ can have the following possibilities:
\begin{align*}
\left \{\begin{matrix}
\text{(a)} :& D_0 = 0 & D_1 = 0 & D_2 = 0 \\ 
\text{(b)} : & D_0 < 0 & D_1 = 0 & D_2 < 0 \\
\text{(c)} : & D_0\leq 0 & D_1>0 & D_2 \leq 0\\
\text{(d)} : & D_0\geq 0 & D_1<0 & D_2 \geq 0 \\
\text{(e)} : & D_0\leq 0 & D_1<0 & D_2\leq 0 \\
\text{(f)} : & D_0\leq 0 & D_1<0 & D_2 \geq 0 \\
\text{(g)} : & D_0\geq 0 & D_1<0 & D_2 \leq 0 
\end{matrix} \right.
\end{align*}
We show that in each case (if it can occur), we can reach to $\left[P,0\right]$ or $\left[0,P\right]$ from $\boldsymbol{x_1}$ without increasing the error probability. \\ 
(a): Lemma~\ref{l7} implies that the two codewords are equal. Therefore the code is not optimal.\\
(b): Using Lemma~\ref{l7} for $D_0, D_2$ and from the definitions of $\boldsymbol{\lambda}$ and $\boldsymbol{\mu}$ given in \eqref{lmu-num}, we obtain that $x_{10}$ and $x_{11}$ are greater than $x_{20}$, $x_{21}$. Then, any linear combination of $x_{10}$ and $x_{11}$ is also greater than the same combination of $x_{20}$, $x_{21}$. This contradicts Lemma~\ref{l7} for $D_1$.  \\
(c),(d): These cases are similar to case (b). \\
(e): Using Lemma~\ref{l7} for $D_0, D_2$ we obtain that $x_{10}$ and $x_{11}$ are greater than or equal to
$x_{20}$ and $x_{21}$. Therefore, $\lambda_i\geq \mu_i$ for all $i$.
Lemma~\ref{l1} then implies that if we set decrease $x_{20},x_{21}$ to zero, the error probability would not increase. Next, if $x_{11}\pi_1 \leq x_{10} \pi_0$, Lemma~\ref{l3} shows that increasing $x_{10}$ to $P$ and decreasing $x_{11}$ to $0$ would not increase the error probability. A similar argument works for $x_{11}\pi_1 \geq x_{10} \pi_0$. \\
(f): We have $x_{10} \pi_0 \geq x_{20} \pi_0$,  $ x_{11}\pi_1 \leq x_{21}\pi_1$ and $x_{11} \pi_0 + x_{10} \pi_1 > x_{20} \pi_0 + x_{21} \pi_1$. If $\pi_0 \leq \pi_1$, we decrease $x_{11}$ to $0$ and increase $x_{10}$ to $P$; this would not increase $\mathbb{P}_{\e}$ by Lemma~\ref{l1}. If $\pi_0 > \pi_1$, we proceed as follows: we vary the values of $\boldsymbol{\lambda}$ in \eqref{lmu-num} in a number of steps such that the error probability does not increase in each step. While the values of $\boldsymbol{\lambda}$ that we obtain in the intermediate steps do not necessarily correspond to codewords $\boldsymbol{x_1}$, the final $\boldsymbol{\lambda}$ that we reach does correspond to codewords $\boldsymbol{x_1}$ as given in the statement of the theorem.
We first decrease $x_{11} \pi_0 + x_{10} \pi_1$ by $x_{11}\left(\pi_0-\pi_1\right)$ to $\left(x_{11}+x_{10}\right) \pi_1=P\pi_1$ and increase $x_{10}\pi_0$  by $x_{11}\left(\pi_0-\pi_1\right)$ to $x_{10}\pi_0+x_{11}\left(\pi_0-\pi_1\right)$. This change does not increase the error since $D_0 \leq 0$ and $D_2 \geq 0$; by Lemma~\ref{l7} we have $x_{10}\geq x_{20}$ and $x_{11}\leq x_{21}$; therefore, \[\frac{x_{10}\pi_0+d}{x_{20}\pi_0+d} \geq \frac{x_{11} \pi_0 + x_{10} \pi_1+d}{x_{20} \pi_0 + x_{21} \pi_1+d}.\] The latter follows the fact that for every positive $a, a', b,$ and $b'$, if $\frac{a}{a'}\leq \frac{b}{b'}$ then $\frac{a}{a'} \leq \frac{a+b}{a'+b'} \leq \frac{b}{b'}$. Hence, Lemma~\ref{l3} ensures that the error probability would not increase. Next, decreasing $x_{11}\pi_1$ to $0$ and then increasing $x_{10}\pi_0+x_{11}\left(\pi_0-\pi_1\right)$ to $\left(x_{10}+x_{11}\right)\pi_0=P\pi_0$ would not increase the error probability by Lemma~\ref{l1}. Therefore, we reach to codeword $\boldsymbol{\lambda}=\left[P\pi_0, P\pi_1, 0\right]$ that corresponds to codeword $\boldsymbol{x_1}=\left[P,0\right]$.\\ 
(g): This case is similar to case (f).
\item[3-] \emph{Either $\begin{pmatrix}
\boldsymbol{x_1}\\
\boldsymbol{x_2}
\end{pmatrix}=\left( \begin{matrix}
x & 0 \\
 0& x'
\end{matrix}
\right)$ or $\begin{pmatrix}
\boldsymbol{x_1}\\
\boldsymbol{x_2}
\end{pmatrix}=\left( \begin{matrix}
0 & x\\
 0& x'
\end{matrix}
\right)$ hold for some $0\leq x' \leq x \leq P$.} For the second case, error probability would not increase if we reduce $x'$ to $0$ and increase $x$ to $P$. For the first one if $ x\pi_1 \geq x' \pi_0$ we can increase $x$ to $P$ and else increase $x'$ to $P$. Hence the  optimal code has to have the following structure for some $0 \leq x \leq P$:
\begin{align*}
\begin{pmatrix}
\boldsymbol{x_1}\\
\boldsymbol{x_2}
\end{pmatrix}=\left( \begin{matrix}
P & 0 \\
 0& x
\end{matrix}
\right). 
\end{align*}
$\hfill  \square $
\end{itemize}

\subsection{Proof of Theorem~\ref{t3}}
\label{five3}

 Take an arbitrary codeword pair $(\boldsymbol{x_1}, \boldsymbol{x_2})$. The error probability of the codeword pair $(\boldsymbol{x_1}, \boldsymbol{x_2})$ is determined by $ \boldsymbol{\lambda}$ and $\boldsymbol{\mu}$ where
\begin{align}
\left(\begin{matrix}
\boldsymbol{\lambda}\\
\boldsymbol{\mu}
\end{matrix}\right)=\left(
\begin{matrix}
\boldsymbol{\pi}*\boldsymbol{x_1}\\
\boldsymbol{\pi}*\boldsymbol{x_2}
\end{matrix}\right)=\left(
\begin{matrix}
\lambda_0 &  \ldots &\lambda_{N+K-2} \\
\mu_0 &  \ldots & \mu_{N+K-2} 
\end{matrix}\right).
\end{align}
The idea of the proof is as follows:
we alter the pair $(\boldsymbol{\lambda}, \boldsymbol{\mu})$ in a number of steps such that (i) in each alteration step, the error probability of the altered pair is less than or equal to  the error probability of the unaltered pair, and (ii) at the end of the alteration steps, we either reach the pair corresponding to the one given in the statement of the theorem, or we reach a pair  whose error probability is greater than or equal to the one given in the statement of the theorem provided that $A$ is sufficiently large.  This shows the optimality of the codeword pair given in the statement of the theorem.

Let $\lambda'_i=\lambda_i \mathbf{1}[\lambda_i\geq \mu_i]$ and $\mu'_i=\mu_i \mathbf{1}[\lambda_i<\mu_i]$. By operations of Lemma~\ref{l1}, we can reach from $\left( \boldsymbol{\lambda}, \boldsymbol{\mu}\right)$ to $\left( \boldsymbol{\lambda'}, \boldsymbol{\mu'}\right)$ without increasing the error probability. However, observe that it may not be possible to express $\left( \boldsymbol{\lambda'}, \boldsymbol{\mu'}\right)$ as $\left(\boldsymbol{\pi}*\boldsymbol{x'_1}, \boldsymbol{\pi}*\boldsymbol{x'_2}\right)$ for some codewords $\boldsymbol{x'_1}$ and $\boldsymbol{x'_2}$. Nonetheless, we can identify two new codewords $\boldsymbol{x'_1}$ and $\boldsymbol{x'_2}$ whose error probability is less than or equal to that of the pair $\left( \boldsymbol{\lambda'}, \boldsymbol{\mu'}\right)$.

If $|\{i:\lambda'_i>0\}|=N+K-1$, we have $\boldsymbol{\mu'}=0$. Since $\boldsymbol \lambda'  \prec_{\w} \boldsymbol \lambda=\boldsymbol{x_1} *\boldsymbol{\pi} \prec_{\w} \boldsymbol{x'_1}*\boldsymbol{\pi} =\boldsymbol{\lambda''}$, where $\boldsymbol{x'_1}=\left[A,A,  \ldots ,A\right]$, utilizing Lemma~\ref{l6}, $\boldsymbol{\lambda'}$ is majorized by $\boldsymbol{\lambda''_1}=\left[\lambda''_0,\lambda''_1,\ldots,\lambda''_t, r_1,0,\ldots,0\right]$ where $\boldsymbol{\lambda''}$ is sorted in the descending order and $0 \leq r_1 \leq \lambda''_{t+1}$.  We can move from $\boldsymbol{\lambda'}$ to $\boldsymbol{\lambda''_1}$ with AR operations using
Lemma~\ref{l5}. Corollary~\ref{corrl1} guarantees that the error probability does not increase during
these AR operations. Using Lemma~\ref{l1} we can increase the
elements of $\boldsymbol{\lambda''_1}$ and reach to $\boldsymbol{\lambda''}=\boldsymbol{x'_1}*\boldsymbol{\pi}$ with power $NA\left(\sum_{i=0}^{K-1} \pi_i\right)$ without increasing the error probability. Therefore, the error probability of $\left( \boldsymbol{\lambda}, \boldsymbol{\mu}\right)$ is greater than or equal to the error probability of $\left( \boldsymbol{\lambda''}=\boldsymbol{\pi}*\boldsymbol{x'_1}, \boldsymbol{\mu'}=\boldsymbol{\pi}*\boldsymbol{x'_2}\right)$ where $\boldsymbol{x'_1}=\left[A,A,  \ldots ,A\right]$ and $\boldsymbol{x'_2}=\left[0,0,  \ldots ,0\right]$. We are done in this case. The case $|\{i:\lambda'_i>0\}|=0$ is similar.

 If $0<|\{i:\lambda'_i>0\}| <N+K-1$, we have $\mu_j>\lambda_j$ for some $j$. Without loss of generality assume that the power of $\boldsymbol{\lambda'}$ is greater than or equal to the power of $\boldsymbol{\mu'}$ (otherwise, we can swap the two). We claim that the power of $\boldsymbol{\lambda'}$ is less than or equal to
 \begin{align*}
 \max_{j\in \{ 0, \ldots, K-1 \} }\left(-A\pi_j +NA\left(\sum_{i=0}^{K-1} \pi_i\right)\right). 
 \end{align*}
 This holds if $x_{1i}=0$ for some $0 \leq i \leq N+K-2$ as the total would be at most $$\left(N-1\right)A\left(\sum_{i=0}^{K-1} \pi_{i}\right),$$ in this case. Thus, assume that $x_{1i}>0$ for all $0 \leq i \leq N+K-2$. 

 Suppose that the first index $j$ where $\mu_j>\lambda_j$ is $\bar j$. The weight of $\boldsymbol{\lambda'}$ equals $\sum_{{i:\lambda'_i \neq 0}}\lambda_i$. We have
\begin{align*}
&\sum_{i:\lambda'_i \neq 0}\lambda_i \\
& \leq \sum_{i=0, {i \neq \bar j}}^{N+K-2} \lambda_i \\
& \overset{\text{(a)}}{\leq} -\lambda_{\bar j} + \left( \sum_{i=0}^{N-1} x_{1i}\right) \left( \sum_{i=0}^{K-1} \pi_i \right)\\ 
&\overset{\text{(b)}}{\leq}
-\lambda_{\bar j}+\left( x_{1\left(\min\{ {\bar{j}, N-1} \}\right)}+\left(N-1\right)A\right) \left( \sum_{i=0}^{K-1} \pi_i \right)  \\
&\overset{\text{(c)}}{\leq}
-x_{{1\left(\min \{{\bar{j}, N-1} \}\right)}} \pi_{\left(\left(\bar{j}-\left(N-1\right)\right)_+\right)} \\
&\quad  -\left(A- x_{1\left(\min\{ {\bar{j}, N-1}\}\right)}\right) \left(\sum_{i=0}^{K-1} \pi_i\right)+NA\left(\sum_{i=0}^{K-1} \pi_i\right)\\
&\leq 
\max_{x \in [0, A]} \left \{ -x\pi_{\left(\left(\bar{j}-\left(N-1\right)\right)_+\right)} -\left(A-x\right) \sum_{i=0}^{K-1} \pi_i \right \}\\
&\qquad +NA\left(\sum_{i=0}^{K-1} \pi_i\right)\\
 & \overset{\text{(d)}}{\leq} 
 -A\pi_{\left(\left(\bar{j}-\left(N-1\right)\right)_+\right)} +NA\left(\sum_{i=0}^{K-1} \pi_i\right)\\
&\leq
\max_{j\in \{ 0, \ldots, K-1 \} }\left(-A\pi_j +NA\left(\sum_{i=0}^{K-1} \pi_i\right)\right).
\end{align*}

 In the above derivation, the step (a) follows from $\boldsymbol{\lambda}=\boldsymbol{x_1}*\boldsymbol{\pi} $ and nonnegativity of
$\boldsymbol{x_1}$ and $\boldsymbol{\pi}$  which imply that the total-power of $\boldsymbol{\lambda}$ (sum of its elements) is not greater than the total-power of $\boldsymbol{x_1}$ times the total-power of $\boldsymbol{\pi}$. Step (b) follows from $x_{1i} \leq A$ for $0 \leq i \leq N-1$ (the peak-power constraint). The step (c) follows from $\boldsymbol{\lambda}=\boldsymbol{x_1}*\boldsymbol{\pi} $ and the nonnegativity of $\boldsymbol{x_1}$ and $\boldsymbol{\pi}$ which imply $ \lambda_i \leq x_{{1\left(\min \{{i, N-1} \}\right)}} \pi_{\left(\left(i-\left(N-1\right)\right)_+\right)}$ for $0\leq i \leq N+K-2$. Finally (d) follows from the fact that a linear function of $x \in [a,b]$ takes its maximum at $a$ or $b$.

Since $\pi_j>0$ for all $j$, we have 
\begin{align}
\max_{j\in \{ 0, \ldots, K-1 \} }\left(-A\pi_j +NA\left(\sum_{i=0}^{K-1} \pi_i\right)\right)<NA\left(\sum_{i=0}^{K-1} \pi_i\right).
\label{opta}
\end{align}
From Lemma~\ref{l9}, we can find the upper bound \eqref{c1-eq2} given at the top of page~\pageref{c1-eq2}, 
on the error probability of the code given in the statement of the theorem (the codewords $\left[A,A,  \ldots ,A\right]$ and $\left[0,0,  \ldots ,0\right]$). We can also find the lower bound \newcounter{storeeqcounter1}
\newcounter{tempeqcounter1}
\setcounter{storeeqcounter1}{\value{equation}}
\addtocounter{equation}{1}
\begin{figure*}[!t]
\normalsize
\setcounter{tempeqcounter1}{\value{equation}}
\begin{align}
\setcounter{equation}{\value{storeeqcounter1}}
\label{c1-eq2}
    &e^{-(N+K-1)\left(\frac {NA\left(\sum_{i=0}^{K-1} \pi_i\right)}{N+K-1}+d\right)}\nonumber\\&\quad\times \left( \small{\frac{e\left(\frac {NA\left(\sum_{i=0}^{K-1} \pi_i\right)}{N+K-1}+d\right)\log\left({\left(\frac {NA\left(\sum_{i=0}^{K-1} \pi_i\right)}{N+K-1}+d\right)}/{d}\right)}{\frac {NA\left(\sum_{i=0}^{K-1} \pi_i\right)}{N+K-1}}} \right)^{{(N+K-1)}/{\log\left({\left(\frac {NA\left(\sum_{i=0}^{K-1} \pi_i\right)}{N+K-1}+d\right)}/{d}\right)}}
 \nonumber\\&=\mathcal{O}\left(e^{-NA\left(\sum_{i=0}^{K-1} \pi_i\right)}\log\left(NA\left(\sum_{i=0}^{K-1} \pi_i\right)\right)^{{(N+K-1)}/{\log\left(NA\left(\sum_{i=0}^{K-1} \pi_i\right)\right)}}\right)
\end{align}
\setcounter{equation}{\value{tempeqcounter1}} 
\hrulefill
\vspace*{4pt}
\end{figure*}
\begin{align*}
 \frac12 e^{-\left(\max_{j\in \{ 0, \ldots, K-1 \} }\left(-A\pi_j +NA\left(\sum_{i=0}^{K-1} \pi_i\right)\right)+(N+K-1)d\right)}
\end{align*}
on  the error probability of the code $(\boldsymbol{\lambda}',\boldsymbol{\mu}')$. For sufficiently large $A$ we obtain \eqref{c1-eq3} given at the top of page~\pageref{c1-eq3},
since the left hand side decays faster than the right hand side to zero as $A$ tends to infinity by  \eqref{opta}. 
Thus, for large values of $A$, the error probability of the pair $\left( \boldsymbol{\lambda'}, \boldsymbol{\mu'}\right)$  is greater than the error probability of the code given in the statement of the theorem.

\newcounter{storeeqcounter2}
\newcounter{tempeqcounter2}
\setcounter{storeeqcounter2}{\value{equation}}
\addtocounter{equation}{1}
\begin{figure*}[!t]
\normalsize
\setcounter{tempeqcounter2}{\value{equation}}
\begin{align}
\setcounter{equation}{\value{storeeqcounter2}}
\label{c1-eq3}
        &e^{-(N+K-1)\left(\frac {NA\left(\sum_{i=0}^{K-1} \pi_i\right)}{N+K-1}+d\right)}\nonumber\\&\quad\times \left( \small{\frac{e\left(\frac {NA\left(\sum_{i=0}^{K-1} \pi_i\right)}{N+K-1}+d\right)\log\left({\left(\frac {NA\left(\sum_{i=0}^{K-1} \pi_i\right)}{N+K-1}+d\right)}/{d}\right)}{\frac {NA\left(\sum_{i=0}^{K-1} \pi_i\right)}{N+K-1}}} \right)^{{(N+K-1)}/{\log\left({\left(\frac {NA\left(\sum_{i=0}^{K-1} \pi_i\right)}{N+K-1}+d\right)}/{d}\right)}}
 \nonumber\\&\leq \frac12 e^{-\left(\max_{j\in \{ 0, \ldots, K-1 \} }\left(-A\pi_j +NA\left(\sum_{i=0}^{K-1} \pi_i\right)\right)+(N+K-1)d\right)},
\end{align}
\setcounter{equation}{\value{tempeqcounter2}} 
\hrulefill
\vspace*{4pt}
\end{figure*}
$\hfill  \square $
\subsection{Proof of Theorem~\ref{t4}}
\label{five4}
Consider two arbitrary codewords  $\boldsymbol{ x_1}$ and $\boldsymbol {x_2}$. Consider the case that $x_{1i}$ and $x_{2i}$ are both positive for some $i$. Let us assume that $x_{1i}\leq x_{2i}$ (the other case is similar). Then, using Lemma~\ref{l1}, reducing $x_{1i}$ to zero does not increase the error probability. This change does also does not increase the power of the codeword $\boldsymbol{x}_1$. Thus, using Lemma~\ref{l1} repeatedly, we can reach to codewords of the following form without increasing the error probability:
\begin{align}
\left(
\begin{matrix}
\boldsymbol{x'_1}\\
\boldsymbol{x'_2} 
\end{matrix}\right)=\left(
 \begin{matrix}
   x'_{10} &  \ldots  & x'_{1\left(n-1\right)}&  0 &   \ldots   & 0 \\
 0&  \ldots  & 0 & x'_{2n} &   \ldots   & x'_{2\left(N-1\right)}
 \end{matrix}\right).
\end{align}
Using Corollary~\ref{corrl1}, we can reach the following codewords without increasing the error probability (with AR operations):
\begin{align}
&\left(
\begin{matrix}
\boldsymbol{x''_1}\\
\boldsymbol{x''_2} 
\end{matrix}\right)\nonumber\\
&\quad =\left(
 \begin{matrix}
   A &  \cdots  & A & x''_{1n'}&  0 &   \cdots   &0& 0  \\
  0&  \ldots  & 0 & 0& A &   \ldots   &A& x''_{2n''} 
 \end{matrix}
 \begin{matrix}
  &0 &  \ldots  & 0\\
   &0 &  \ldots  & 0 
 \end{matrix}\right)
 \label{eqnx2}
\end{align}
for some positive $x''_{1n'}$ and $x''_{2n''}$ and $n'' \leq \left(N-1\right)$. Without loss of generality assume that $$\sum_{i=0}^{N-1}{x''_{1i}}\geq \sum_{i=0}^{N-1}{x''_{2i}},$$
else we can switch the two codewords.
Using Lemma~\ref{l1}, increasing 
$x''_{1n'}$ or $x''_{2n''}$ would decrease the error probability of the code. Thus, $x''_{1n'}$ or $x''_{2n''}$ can be increased as long as the  power constraint allow it. This implies that $x''_{1n'}$ and $x''_{2n''}$ can be increased to either $A$ or $A\{\beta\}$. 

Two cases are possible: if $\sum_{i=0}^{N-1}{x''_{1i}}=A\beta$, then the error probability of the codeword given in \eqref{eqnx2} is greater than equal error probability of the codes given in the statement of the theorem by Lemma~\ref{l1}. Else if $\sum_{i=0}^{N-1}{x''_{1i}}<A\beta$, since  $x''_{1n'}$ is either $A$ or $A\{\beta\}$, we deduce that
$$A\beta-\sum_{i=0}^{N-1}{x''_{1i}}\geq A\Gamma$$
where $\Gamma>0$ is a constant defined as follows:
\begin{align}
  \Gamma\triangleq \begin{cases} \{\beta\} &\mbox{if } \{\beta\} > 0 \\
1 & \mbox{if } \{\beta\} = 0. \end{cases}\label{eqnNAAAA}
\end{align}
From Lemma~\ref{l9}, we can find the upper bound 
\begin{align*}
&e^{-n\left(\frac {A\beta}{n}+d\right)}\\ 
&\quad\times\left( \small{\frac{e\left(\frac {A\beta}{n}+d\right)\log\left({\left(\frac {A\beta}{n}+d\right)}/{d}\right)}{\frac {A\beta}{n}}} \right)^{{n}/{\log\left({\left(\frac {A\beta}{n}+d\right)}/{d}\right)}}
\end{align*}
on the error probability of the code given in the statement of the theorem, where
{\begin{align}
  n\triangleq \begin{cases} \lfloor \beta \rfloor +1 &\mbox{if } \{\beta\} > 0 \\
\beta & \mbox{if } \{\beta\} = 0. \end{cases}
\end{align}}
Similarly, from Lemma~\ref{l9}, we can find the lower bound
$$\frac12 e^{-\left(\sum_{i=0}^{N-1}{x''_{1i}}+Nd\right)}\geq \frac12 e^{-\left(A(\beta-\Gamma)+Nd\right)}$$
on the error probability of the code in \eqref{eqnx2}. 
This proves that the code given in the statement of the theorem has a smaller probability if $A$ is sufficiently large such that
\begin{align}
    &\frac12 e^{-\left(A(\beta-\Gamma)+Nd\right)}\nonumber\\
    & \quad \geq e^{-n\left(\frac {A\beta}{n}+d\right)}\nonumber\\
    & \qquad \times\left( \small{\frac{e\left(\frac {A\beta}{n}+d\right)\log\left({\left(\frac {A\beta}{n}+d\right)}/{d}\right)}{\frac {A\beta}{n}}} \right)^{{n}/{\log\left({\left(\frac {A\beta}{n}+d\right)}/{d}\right)}}.\label{eqnNNNNNNNN}
\end{align}
Observe that for large enough $A$, the above equation holds since the right hand side of of the order
\begin{align*}
\mathcal{O}\left(e^{-A\beta}\log\left(A\beta\right)^{{n}/{\log\left(A\beta\right)}}\right)
\end{align*}
as $A$ tends to infinity while $n<N$ is fixed, which vanishes faster than the left hand side of  \eqref{eqnNNNNNNNN} as $A$ tends to infinity.
$\hfill  \square $

\section{Conclusion}
\label{six}
In this paper, we studied the transmission of a bit over a discrete Poisson channel with memory  under a peak-power or total-power  constraint. For the case of having a  total-power  constraint, the optimal codewords for $N>K$ (code length greater than memory length) are derived. An simple ``bursty" code was shown to be optimal in this case. Interestingly, the codewords do not depend on the channel memory coefficients, meaning that the knowledge of channel memory coefficients is not necessary at the transmitter if $N>K$. The problem seems to be difficult to solve when $N \leq K$. In particular, Example~\ref{example6} shows that the channel memory coefficients affect the structure of the optimal code in this case. This special case is left as a future work.  When the  peak-power  constraint is imposed, an on/off
keying  strategy  is shown to be optimal in the high-power regime  regardless of the values of the channel coefficients. Finally, we also investigated the case where both the  peak-power  constraint and the  total-power  constraint are imposed. We only provided a result for a channel without memory in the  high-power regime. It does not seem easy to explicitly identify the optimal codewords for the most general case.

\appendices

\section{Some Lemmas}
\label{appendixa}
Take an arbitrary pair $\left(\boldsymbol{\lambda},\boldsymbol{\mu}\right)$ of nonnegative real vectors of length $N+K-1$. 
Let 
\begin{align}
    \mathcal{A}&=\{0\leq i\leq N+K-2: \lambda_i\geq \mu_i\},\label{defTildeA1}\\
    \mathcal{A}^c&=\{0,1,\ldots, N+K-2\} \setminus \mathcal A, \label{defTildeA2}
\end{align}
and 
\begin{align}
a_i&=\log \frac{\lambda_i+d}{\mu_i+d}, \ i \in \{0,1,\ldots, N+K-2\},\label{aid}
\\ b&=\sum_{i=0}^{N+K-2} \left(\lambda_i-\mu_i\right). \label{bd}
\end{align}
 Then,  the decision rule given in \eqref{eqn:LMu} can be written as follows.       
\begin{align}
\sum_{i=0}^{N+K-2}a_iY_i=\sum_{i \in \mathcal{A}} |a_i| Y_i   - \sum_{i \in \mathcal{A}^c} |a_i| Y_i
  \, \underset{2}{\overset{1}{\gtreqless}}  \, b.\label{eqn:DR1}
\end{align}
\begin{lemma}
\label{l1}
$\mathbb{P}_{\e}\left(\boldsymbol{\lambda},\boldsymbol{\mu}\right)\geq \mathbb{P}_{\e}\left(\boldsymbol{\lambda'},\boldsymbol{\mu'}\right)$ if 
$\lambda_i\leq \lambda'_i$, $\mu_i\geq \mu'_i$ for $i \in \mathcal{A}$, and $\lambda_i\geq \lambda'_i$, $\mu_i\leq \mu'_i$ for $i \in \mathcal{A}^c$, where $\mathcal{A}$ is defined in \eqref{defTildeA1} and $ \mathbb{P}_{\e}\left(\boldsymbol{\alpha},\boldsymbol{\beta}\right)$ is the error probability under the optimal decision rule for the pair $\left(\boldsymbol{\alpha},\boldsymbol{\beta}\right)$.
\end{lemma}
\begin{proof} $ \mathbb{P}_{\e}\left(\boldsymbol{\lambda'},\boldsymbol{\mu'}\right)$ is the error probability under the optimal decision rule for the pair $\left(\boldsymbol{\lambda'},\boldsymbol{\mu'}\right)$. If we use a different decision rule, the resulting error probability will be greater than or equal to $ \mathbb{P}_{\e}\left(\boldsymbol{\lambda'},\boldsymbol{\mu'}\right)$. We show that if we use the optimal decision rule of $\left(\boldsymbol{\lambda},\boldsymbol{\mu}\right)$ for the pair $\left(\boldsymbol{\lambda'},\boldsymbol{\mu'}\right)$, the error probability will be less than or equal to $\mathbb{P}_{\e}\left(\boldsymbol{\lambda},\boldsymbol{\mu}\right)$. This shows that 
$\mathbb{P}_{\e}\left(\boldsymbol{\lambda},\boldsymbol{\mu}\right)\geq \mathbb{P}_{\e}\left(\boldsymbol{\lambda'},\boldsymbol{\mu'}\right)$.

Let $a_i$ and $b$ be defined as in \eqref{aid} and \eqref{bd}. 
Then, $\mathbb{P}_{\e}\left(\boldsymbol{\lambda},\boldsymbol{\mu}\right)$ is the average of error probability under the hypothesis that message $B=1$ or message $B=2$ are transmitted. In other words, $\mathbb{P}_{\e}\left(\boldsymbol{\lambda},\boldsymbol{\mu}\right)$ equals 
\begin{align*}
 &\frac12 \mathbb{P}\left(
\sum_{i \in \mathcal{A}} |a_i| Y_i   - \sum_{i \in \mathcal{A}^c} |a_i| Y_i
< b \middle|  Y_i\sim \mathsf{Poisson}\left(\lambda_i+d\right)\right)\nonumber
\\&  + \frac12 \mathbb{P} \left(
\sum_{i \in \mathcal{A}} |a_i| Y_i   - \sum_{i \in \mathcal{A}^c} |a_i| Y_i
\geq b \middle|  Y_i\sim \mathsf{Poisson}\left(\mu_i+d\right)\right)
\end{align*}
where $Y_i$'s are mutually independent under each hypothesis.

With the decision rule given by $a_i$ and $b$, the error probability for the pair $\left(\boldsymbol{\lambda'},\boldsymbol{\mu'}\right)$, $\tilde{\mathbb{P}}_{\e}\left(\boldsymbol{\lambda'},\boldsymbol{\mu'}\right)$ is equal to the following.
\begin{align}
& \frac12 \mathbb{P}\left(
\sum_{i \in \mathcal{A}} |a_i| Y'_i   - \sum_{i \in \mathcal{A}^c} |a_i| Y'_i
< b \middle| Y'_i\sim \mathsf{Poisson}\left(\lambda'_i+d\right)\right)\nonumber\\
&+\frac12 \mathbb{P}\left(
\sum_{i \in \mathcal{A}} |a_i| Y'_i   - \sum_{i \in \mathcal{A}^c} |a_i| Y'_i
\geq b  \middle|  Y'_i\sim \mathsf{Poisson}\left(\mu'_i+d\right)\right)\label{ee}
\end{align}
To show that the expression given in \eqref{ee} is less than or equal to $\mathbb{P}_{\e}\left(\boldsymbol{\lambda},\boldsymbol{\mu}\right)$, it suffices to show that
\begin{align*} &\mathbb{P}\left(
\sum_{i \in \mathcal{A}} \right.  \left. |a_i| Y_i   - \sum_{i \in \mathcal{A}^c} |a_i| Y_i
< b \middle|  Y_i\sim \mathsf{Poisson}\left(\lambda_i+d\right)\right)\\ &\geq 
\mathbb{P}\left(
\sum_{i \in \mathcal{A}} |a_i| Y'_i   - \sum_{i \in \mathcal{A}^c} |a_i| Y'_i
< b \middle|  Y'_i\sim \mathsf{Poisson}\left(\lambda'_i+d\right)\right)
\end{align*}
and
\begin{align*}
&\mathbb{P}
\left( 
\sum_{i  \in \mathcal{A}} \right.  \left. |a_i| Y_i  - \sum_{i \in \mathcal{A}^c} |a_i| Y_i
\geq b \middle|  Y_i\sim \mathsf{Poisson}\left(\mu_i+d\right)\right)\\ 
&  \geq 
\mathbb{P}\left(
\sum_{i \in \mathcal{A}} |a_i| Y'_i   - \sum_{i \in \mathcal{A}^c} |a_i| Y'_i
\geq b \middle| Y'_i\sim \mathsf{Poisson}\left(\mu'_i+d\right)\right).
\end{align*}
We prove the first equation. The proof for the second one is similar. Let $ Y_i\sim \mathsf{Poisson}\left(\lambda_i+d\right)$ and $Y'_i\sim \mathsf{Poisson}\left(\lambda'_i+d\right)$ be independent Poisson random variables. It suffices to show that
\begin{align} 
&\mathbb{P}\left(
\sum_{i \in \mathcal{A}} |a_i| Y_i   - \sum_{i \in \mathcal{A}^c} |a_i| Y_i
< b\right)
\nonumber\\
&\quad\geq 
\mathbb{P}\left(
\sum_{i \in \mathcal{A}} |a_i| Y'_i   - \sum_{i \in \mathcal{A}^c} |a_i| Y_i
< b\right)\label{eqn:1a1}
\\&\quad\geq\mathbb{P}\left(
\sum_{i \in \mathcal{A}} |a_i| Y'_i   - \sum_{i \in \mathcal{A}^c} |a_i| Y'_i
< b\right),\label{eqn:1b1}
\end{align}
To show \eqref{eqn:1a1}, let ${Z}_i\sim \mathsf{Poisson}\left(\lambda'_i-\lambda_i\right)$ for $i\in\mathcal{A}$ be mutually independent of each other and of previously defined variables. Let
$\tilde{Y}_i=Y_i+{Z}_i$ for $i\in\mathcal{A}$. Observe that $\tilde Y_i\sim \mathsf{Poisson}\left(\lambda'_i+d\right)$ has the same distribution as $Y'_i$. Since $Z_i\geq 0$, 
\begin{align} 
&\mathbb{P}\left(
\sum_{i \in \mathcal{A}} |a_i| Y_i   - \sum_{i \in \mathcal{A}^c} |a_i| Y_i
< b\right)
\nonumber\\&\quad\geq
 \mathbb{P}\left(\sum_{i \in \mathcal{A}} |a_i| Z_i +
\sum_{i \in \mathcal{A}} |a_i| Y_i   - \sum_{i \in \mathcal{A}^c} |a_i| Y_i
< b\right)
\\&\quad= \mathbb{P}\left(
\sum_{i \in \mathcal{A}} |a_i| \tilde Y_i   - \sum_{i \in \mathcal{A}^c} |a_i| Y_i
< b\right)
\\&\quad= \mathbb{P}\left(
\sum_{i \in \mathcal{A}} |a_i|  Y'_i   - \sum_{i \in \mathcal{A}^c} |a_i| Y_i
< b\right),
\end{align}
where the last step utilizes the fact that $\tilde Y_i$ has the same distribution as $Y'_i$. 
The proof of \eqref{eqn:1b1} is similar and follows by defining ${Z}_i\sim \mathsf{Poisson}\left(\lambda_i-\lambda'_i\right)$  for $i\in\mathcal{A}^c$ and
$\tilde{Y}_i=Y'_i+{Z}_i$ for $i\in\mathcal{A}^c$.
\end{proof}

\begin{lemma}
\label{l3}
Take some pair $\left(\boldsymbol{\lambda},\boldsymbol{\mu}\right)$. 
Let $a_i$ be defined as in \eqref{aid}. Assume that $a_j\geq a_k$. Let $\boldsymbol{\lambda'}$ be equal to  $\boldsymbol{\lambda}$ in all indices except for indices $j,k$, \emph{i.e.,} $\lambda'_i=\lambda_i$ for $i\notin\{j,k\}$. Furthermore, 
$\lambda'_j=\lambda_j+\epsilon$ and $\lambda'_k=\lambda_k-\epsilon$ for some $\epsilon\leq \lambda_k$. Then, 
$\mathbb{P}_{\e}\left(\boldsymbol{\lambda},\boldsymbol{\mu}\right)\geq \mathbb{P}_{\e}\left(\boldsymbol{\lambda'},\boldsymbol{\mu}\right)$.
\end{lemma}
\begin{corollary}\label{corrl1} Let $\boldsymbol{\lambda'}$ be the result of applying an AR operation on indices $j,k$ of $\boldsymbol{\lambda}$, and $\mu_{j}=\mu_{k}$ then $\mathbb{P}_{\e}\left(\boldsymbol{\lambda},\boldsymbol{\mu}\right)\geq \mathbb{P}_{\e}\left(\boldsymbol{\lambda'},\boldsymbol{\mu}\right)$ since the condition  $a_j\geq a_k$ is equivalent with  $\lambda_j\geq \lambda_k$.
\end{corollary}
\begin{proof} As in the proof of Lemma~\ref{l1}, it suffices to show that if we use the optimal decision rule of $\left(\boldsymbol{\lambda},\boldsymbol{\mu}\right)$ for the pair $\left(\boldsymbol{\lambda'},\boldsymbol{\mu}\right)$, the error probability will be less than or equal to $\mathbb{P}_{\e}\left(\boldsymbol{\lambda},\boldsymbol{\mu}\right)$. More specifically, let $a_i$ and $b$ be defined as in \eqref{aid} and \eqref{bd}.  We use the decision rule given in \eqref{eqn:DR1} for the pair $\left(\boldsymbol{\lambda'},\boldsymbol{\mu}\right)$. With this decision rule, the error probability under the hypothesis $B=2$ is the same for the pairs  $\left(\boldsymbol{\lambda},\boldsymbol{\mu}\right)$ and  $\left(\boldsymbol{\lambda'},\boldsymbol{\mu}\right)$. Therefore, it remains to show that
\begin{align*} &\mathbb{P}\left(
\sum_{i=0}^{N+K-2}a_i Y_i
< b\right)\geq\mathbb{P}\left(
\sum_{i=0}^{N+K-2}a_i Y'_i
< b\right).
\end{align*}
where  $ Y_i\sim \mathsf{Poisson}\left(\lambda_i+d\right)$ and $Y'_i\sim \mathsf{Poisson}\left(\lambda'_i+d\right)$ are mutually independent Poisson random variables.
{ We now use the following fact about Poisson random variables: given an arbitrary random variable $W\sim \mathsf{Poisson}(\alpha_1+\alpha_2)$, we can decompose $W$ as $W=W_1+W_2$ (with probability one) where 
$W_1\sim \mathsf{Poisson}(\alpha_1)$ and $W_2\sim \mathsf{Poisson}(\alpha_2)$ are independent. Random variables $W_1$ and $W_2$ can be constructed from $W$ using the thinning property of the Poisson random variable.
}
Using this fact, we can find mutually independent random variables $\tilde{Y}_j\sim\mathsf{Poisson}\left(\lambda_i+d\right)$ and $Z_j\sim\mathsf{Poisson}\left(\epsilon\right)$ such that $Y'_j=\tilde{Y}_j+Z_j$. Since we assumed that $Y'_i$ for $0\leq i\leq N+K-2$ are mutually independent random variables, the random variables $\tilde{Y}_j$ and $Z_j$ can be also assumed to be independent of $Y'_i$ for all $i\neq j$. 

Next,
\begin{align*}
a_jY'_j+a_kY'_k &=a_j\left(\tilde{Y}_j+Z_j\right)+a_kY'_k \\
&=a_j\tilde{Y}_j+a_k\left(Y'_k+Z_j\right)+\left(a_j-a_k\right)Z_j\\
&\geq a_j\tilde{Y}_j+a_k\left(Y'_k+Z_j\right).
\end{align*}
Let $\tilde{Y}_k=Y'_k+Z_j \sim\mathsf{Poisson}\left(\lambda_k+d\right)$, and $\tilde{Y}_i=Y_i$ for all $i\notin\{j,k\}$. By the above equation, with probability one we have
$$\sum_{i=0}^{N+K-2}a_iY'_i \geq \sum_{i=0}^{N+K-2}a_i\tilde Y_i $$
Therefore,
\begin{align*} &\mathbb{P}\left(
\sum_{i=0}^{N+K-2}a_i \tilde{Y}_i
< b\right)\geq\mathbb{P}\left(
\sum_{i=0}^{N+K-2}a_i Y'_i
< b\right).
\end{align*}
The proof is finished by noting that $\tilde{Y}_i$'s are mutually independent and have the same distribution as $Y_i$'s. 
\end{proof}

\begin{lemma}
\label{l6}
If $\boldsymbol{\lambda} \prec_{\w} \boldsymbol{\pi}$ for two decreasing sequences $\boldsymbol{\lambda}$ and $\boldsymbol{\pi}$ of length $N$ (\emph{i.e.,} $\lambda_i\geq \lambda_j$ for $i\leq j$). Then $\boldsymbol{\lambda} \prec \boldsymbol{\pi'}$ where
 \begin{align*}
 \boldsymbol{\pi'} = [\pi_0,\pi_1,\ldots, \pi_t,\sum_{i=0}^{N-1}\lambda_i-\sum_{j=0}^{t}\pi_j, 0,\ldots,0]
 \end{align*}
for some $t\in\{0,1, \ldots , N-2\}$ satisfying
 \begin{align}
0\leq \sum_{i=0}^{N-1}\lambda_i-\sum_{j=0}^{t}\pi_j \leq \pi_{t+1}.\label{eqn1nn}
 \end{align}
\end{lemma}
\begin{proof}Since $0\leq \sum_{i=0}^{N-1}\lambda_i \leq \sum_{j=0}^{N-1}\pi_j$, there is some $t$ satisfying
 \begin{align}
\sum_{j=0}^{t}\pi_j\leq \sum_{i=0}^{N-1}\lambda_i \leq \sum_{j=0}^{t+1}\pi_j.
 \end{align}
Therefore \eqref{eqn1nn} holds. One can directly verify by inspection that $\boldsymbol{\lambda} \prec \boldsymbol{\pi'}$.
\end{proof}

\begin{lemma}
\label{l8}Let $C$ be a code $\left(\boldsymbol{\lambda}, \boldsymbol{\mu}\right)$ with length $N$ as follows:
\begin{align*}
\begin{pmatrix}
\boldsymbol{\lambda}\\
\boldsymbol{\mu}
\end{pmatrix}= \Bigg( \overbrace{ \begin{matrix}
A & A &  \cdots  & A  \\     
 0& 0& \cdots  &0                            
\end{matrix}}^N  \Bigg) .
\end{align*}

The error probability of the ML receiver over a Poisson channel with dark noise $d>0$  is of order $\mathcal{O}\left(e^{-NA}\right)$ for large values of $A$. More precisely, if $A\geq 2d$ we have
\begin{align}
\label{el1}
&\frac{1}{2} e ^ {-N\left(A+d\right)}\nonumber\\ &\ \leq \mathbb{P}_{\e}\left(C\right)\nonumber\\&\ \leq  e^{-N\left(A+d\right)} \left( \small{\frac{e\left(A+d\right)\log\left({\left(A+d\right)}/{d}\right)}{A}} \right)^{{NA}/{\log\left({\left(A+d\right)}/{d}\right)}}.
\end{align}
\end{lemma}
\begin{proof}
The error probability for ML receiver is derived as follows:
\begin{align}\mathbb{P}_{\e}\left(C\right)=\frac12\mathbb{P}\Big(Y^{[N\left(A+d\right)]} \leq M\Big)+\frac12\mathbb{P}\Big(Y^{[Nd]} >M\Big),\label{eqn92}\end{align}
where $Y^{[\alpha] }$ denotes a Poisson random variable with mean $\alpha $ and $M=\frac{NA}{\log\left(\left(A+d\right)/d\right)} \geq 0$. Therefore,
$$ \mathbb{P}_{\e}^{C\left(N\right)}\geq\frac12\mathbb{P}\Big(Y^{[N\left(A+d\right)]}=0\Big)=\frac{1}{2} e ^ {-N\left(A+d\right)}.$$ On the other hand, we can bound the error probability from above by utilizing the Chernoff bound on both terms of the error probability in \eqref{eqn92}. 
Since $N\left(A+d\right) \geq M$, we have (see \cite[Eq. 4.1]{alice})
\begin{align}
\label{e1}
\mathbb{P}\Big(Y^{[N(A+d)]} \leq  M\Big) \leq { \frac{e^{-N\left(A+d\right)}\left(eN (A+d)\right)^M}{M^M}}=u_{1},
\end{align}
where 
\begin{align*}
&u_{1}\\
&\ =e^{-N\left(A+d\right)}  \left( \small{\frac{e\left(A+d\right)\log\left({\left(A+d\right)}/{d}\right)}{A}} \right)^{{NA}/{\log\left({\left(A+d\right)}/{d}\right)}}.
\end{align*}
Next, the condition $A\geq 2d$ implies $Nd \leq M$ and we can also write the following Chernoff bound (see \cite[Eq. 4.4]{alice})
\begin{align}
\label{eee1111}
\mathbb{P}\Big(Y^{[Nd]} >  M \Big) \leq { \frac{e^{-Nd}\left(eNd \right)^M}{M^M}}=u_{2}.
\end{align}
One can directly verify that  \begin{align}
\dfrac{  u_{1}}{u_2}=\dfrac{ \frac{e^{-N\left(A+d\right)}\left(eN (A+d)\right)^M}{M^M}}{ \frac{e^{-N\left(d\right)}\left(eNd \right)^M}{M^M}}=e^{-NA}\left( \dfrac{A+d}{d} \right)^M,
\end{align} where $M=\frac{NA}{\log\left(\left(A+d\right)/d\right)}$. Therefore we have $u_1=u_2$.

Hence these two upper bounds imply the upper bound in \eqref{el1}.
 \end{proof} 

\begin{lemma}
\label{l9} 

Take some arbitrary blocklength $N$ and a code $C$ with the following structure
\begin{align}\left(
\begin{matrix}
\boldsymbol{\lambda} \\
 \boldsymbol{\mu}
\end{matrix}\right)=\left(
\begin{matrix}
  \lambda_0 & \ldots &\lambda_{n-1}&  0&  \ldots & 0 \\
  0 &  \ldots  & 0 & \mu_n &  \ldots  & \mu_{\small N-1}
\end{matrix}\right),
\end{align}
with $\sum_{i=0}^{N-1} \lambda_i  \geq \sum_{i=0}^{N-1} \mu_i$. Define $P \triangleq \sum_{i=0}^{N-1} \lambda_i$. 
Then,
\begin{align*}
 &\frac12 e^{-\left(P+Nd\right)}\\ &\ \leq \mathbb{P}_{\e}\left(C\right)\\& \ < e^{-n\left(\frac {P}{n}+d\right)} \left( \small{\frac{e\left(\frac {P}{n}+d\right)\log\left({\left(\frac {P}{n}+d\right)}/{d}\right)}{\frac {P}{n}}} \right)^{{n}/{\log\left({\left(\frac {P}{n}+d\right)}/{d}\right)}}
 \\&\ =\mathcal{O}\left(e^{-P}\log\left(P\right)^{{n}/{\log\left(P\right)}}\right)
\end{align*}
where $\mathcal{O}(\cdot)$ refers to the case in which $P$ tends to infinity while $n<N$ is fixed.

\end{lemma}
\begin{proof}The probability of error is at least $1/2$ times the probability of error when the first message ($\boldsymbol{\lambda}$) is transmitted. From \eqref{eqn:LMu} and the assumption that $\sum_{i=0}^{N-1} \lambda_i  \geq \sum_{i=0}^{N-1} \mu_i$, the threshold in the decision rule is nonnegative. Thus,
\begin{align}
\label{lf}
\mathbb{P}_{\e}\left(C\right) &\geq \frac12 \mathbb{P}\left(\sum_{i=0}^{N-1}  a_i Y^{[\lambda_i+d]}_i =0\right)\nonumber \\
&=\frac12 e^{-\sum_{i=0}^{N-1} \left(\lambda_i+d\right) }\nonumber \\
&=\frac12 e^{-\left(P+Nd\right)},
\end{align}
 where $Y^{[\alpha]}$ denotes a Poisson random variable with mean $\alpha$.
 By deleting the $\left(n+1\right)$-th to $N$-th letters of two codewords the error probability will increase. Therefore $\mathbb{P}_{\e}\left(C\right) \leq \mathbb{P}_{\e}\left(C'\right)$ where $C'$ has the following structure.
\begin{align}
\left(
\begin{matrix}
\boldsymbol{\lambda'}\\
 \boldsymbol{\mu'}
\end{matrix}
\right)=\left(
\begin{matrix}
  \lambda_0 &  \ldots  & \lambda_{n-1}\\
 0&  \ldots  & 0 
 \end{matrix}\right)
\end{align}
Let $\boldsymbol{\lambda''}=\left[P/n, P/n,  \ldots , P/n\right]$ be a sequence of length $n$ with total power $P$, and set $ \boldsymbol{\mu''}= \boldsymbol{\mu'}$. Let $C''$ a code with $\left( \boldsymbol{\lambda''},  \boldsymbol{\mu''}\right)$. Since $ \boldsymbol{\lambda''}\prec \boldsymbol{\lambda'}$, from Lemma~\ref{l5}, with the Robin Hood operations on codeword $\boldsymbol{\lambda'}$ we can reach $\boldsymbol{\lambda''}$. Therefore Corollary~\ref{corrl1}  shows that $\mathbb{P}_{\e}\left(C'\right) \leq \mathbb{P}_{\e}\left(C''\right)$. The result then follows from Lemma~\ref{l8} which implies \begin{align*}
&\mathbb{P}_{\e}\left(C''\right)\nonumber\\& \ \leq 
e^{-n\left(\frac {P}{n}+d\right)} \left( \small{\frac{e\left(\frac {P}{n}+d\right)\log\left({\left(\frac {P}{n}+d\right)}/{d}\right)}{\frac {P}{n}}} \right)^{{n}/{\log\left({\left(\frac {P}{n}+d\right)}/{d}\right)}}\nonumber
\\& \ =\mathcal{O}\left(e^{-P}\log\left(P\right)^{{n}/{\log\left(P\right)}}\right).
&\qedhere
\end{align*}
\end{proof}

\section{Necessary Conditions On an Optimal Code}
\label{AppendixB}
Let $\left(\boldsymbol{x_1},  \boldsymbol{x_2}\right)$ be optimal codewords. Then the error probability of the code should not decrease if we perturb $\boldsymbol{x_1}$ as follows: 
\begin{align}
\label{sss}
 \boldsymbol{x_1} \longrightarrow \boldsymbol{x_1}+\epsilon \boldsymbol{s}
\end{align}
where $\epsilon\geq 0$ and $\boldsymbol{s}$ is any arbitrary sequence of real numbers such that $s_i\geq 0$ whenever 
$x_{1i}=0$, and $ \boldsymbol{x_1}+\epsilon \boldsymbol{s}$ satisfies the power constraints (if any). Equivalently,
\begin{align*}
 &\boldsymbol{\lambda}=\boldsymbol{x_1}*\boldsymbol{\pi} \longrightarrow \boldsymbol{\lambda}+\epsilon {\boldsymbol{s}*\boldsymbol{\pi}}= \boldsymbol{\lambda}+\epsilon \boldsymbol{\zeta} .
\end{align*}
where $\boldsymbol{\zeta} \triangleq \boldsymbol{s}*\boldsymbol{\pi} $. 

We use this idea to show the following necessary condition on the optimal code:
\begin{lemma}
\label{l7}
Take an optimal code and let $\sum_{j=0}^{N+K-2} a_j y_j \underset{2}{\overset{1}{\gtreqless}} b$ be the optimal decision rule for this code. Then  there exists real numbers $D_i$ for $i\in \{0,1, \ldots, N+K-2\}$ such that the following holds: 
\begin{itemize}
\item For any $i\in\{0,1, \ldots, N+K-2\}$, we have
$D_i < 0   \iff  \lambda_i > \mu_i$ and $D_i =0   \iff  \lambda_i = \mu_i$.
\item There is some $\nu\leq 0$ such that for any $j\in\{0,1, \ldots, N-1\}$, if $x_{1j}>0$ then $\sum_{i=0}^{K-1} D_{j+i} \pi_i=\nu$; also, if $x_{1j}=0$ then $\sum_{i=0}^{K-1} D_{j+i} \pi_i\geq \nu$.
Furthermore, if the power of the codeword $\boldsymbol{x_1}$ is strictly less than the  total-power  budget $P$, we have $\nu=0$. 
\end{itemize}
\end{lemma}
\begin{proof}
Let us use the perturbation argument. From the optimality of codewords and using the same decision rule for the perturbed codewords, the error probability should not decrease under the perturbation. That is, $\frac{d}{d\epsilon}\mathbb{P}_{\e}\geq 0$ at $\epsilon=0$. Since we fix the decoding rule, the  probability of making an error when bit $B=2$ is transmitted is unchanged. Thus, we have 
\begin{align*}
&\frac{d}{d\epsilon}\mathbb{P}_{\e}\\
&  =
\frac{d}{d\epsilon}\frac 12 \mathbb{P}\left(\mathsf{error}|B=1\right)\\
& =
\frac{1}{2}\frac{d}{d\epsilon}  \left( \sum_{\boldsymbol{y}:\sum a_j y_j < b} \prod_{i=0}^{N+K-2} e^{-\left(\lambda_i+d+\epsilon \zeta_i\right)} \frac{\left(\lambda_i+d+\epsilon \zeta_i\right)^{y_i}}{y_i !} \right).\\
\end{align*}
Let 
$$F(\epsilon,\boldsymbol{y})= \prod_{i=0}^{N+K-2} e^{-\left(\lambda_i+d+\epsilon \zeta_i\right)} \frac{\left(\lambda_i+d+\epsilon \zeta_i\right)^{y_i}}{y_i !}$$
and for $0 \leq i\leq N+K-2$,
 \begin{align*}
 f_i(\epsilon,\boldsymbol{y}) = e^{-\left(\lambda_i+d+\epsilon \zeta_i\right)} \frac{\left(\lambda_i+d+\epsilon,  \zeta_i\right)^{y_i}}{y_i !}. 
 \end{align*}
 We have
 \begin{align*}
&\frac{\partial }{\partial\epsilon} f_i(\epsilon,\boldsymbol{y})\\
& \  =\frac{\partial }{\partial\epsilon}\left( e^{-(\lambda_i+d+\epsilon \zeta_i)} \frac{(\lambda_i+d+\epsilon \zeta_i)^{y_i}}{y_i !} \right) \\
&\ = -\zeta_i e^{-(\lambda_i+d+\epsilon \zeta_i)} \frac{(\lambda_i+d+\epsilon \zeta_i)^{y_i}}{y_i !}\\
& \ \quad + \zeta_i y_i e^{-(\lambda_i+d+\epsilon \zeta_i)} \frac{(\lambda_i+d+\epsilon \zeta_i)^{y_i -1}}{y_i !}\\
&\ = {e^{-(\lambda_i+d+\epsilon \zeta_i)} \frac{(\lambda_i+d+\epsilon \zeta_i)^{y_i}}{y_i !}} \left(    \frac{\zeta_i y_i}{\lambda_i+d+\epsilon \zeta_i}   -\zeta_i  \right)\\
&\ =f_i(\epsilon,\boldsymbol{y})\left(    \frac{\zeta_i y_i}{\lambda_i+d+\epsilon \zeta_i}   -\zeta_i  \right).
\end{align*}
Since  $F(\epsilon,\boldsymbol{y})= \prod_i f_i(\epsilon,\boldsymbol{y})$, we have $$\frac{\partial }{\partial\epsilon} F(\epsilon,\boldsymbol{y}) = \sum_i \left(\frac{\partial }{\partial\epsilon} f_i(\epsilon,\boldsymbol{y}) \prod_{j \neq i} f_j(\epsilon,\boldsymbol{y}) \right).$$ Therefore we have
 \begin{align*}
&\frac{\partial }{\partial\epsilon} F(\epsilon,\boldsymbol{y})\\
&\quad = \sum_{i=0}^{N+K-2} \left(\frac{\partial }{\partial\epsilon} f_i(\epsilon,\boldsymbol{y}) \prod_{j=0, j \neq i}^{N+K-2} f_j(\epsilon,\boldsymbol{y}) \right) \\
&\quad = \sum_{i=0}^{N+K-2}   \left( \left(    \frac{\zeta_i y_i}{\lambda_i+d+\epsilon \zeta_i}   -\zeta_i  \right) {\prod_{j=0}^{N+K-2} f_j(\epsilon,\boldsymbol{y})}\right)\\
&\quad = F(\epsilon,\boldsymbol{y}) \sum_{i=0}^{N+K-2} \left(    \frac{\zeta_i y_i}{\lambda_i+d+\epsilon \zeta_i}-\zeta_i  \right)\\
& \quad = \left( \prod_{i=0}^{N+K-2} e^{-(\lambda_i+d+\epsilon \zeta_i)} \frac{(\lambda_i+d+\epsilon \zeta_i)^{y_i}}{y_i !} \right)\\
& \qquad \times \left( \sum_{i=0}^{N+K-2} \left(    \frac{\zeta_i y_i}{\lambda_i+d+\epsilon \zeta_i}-\zeta_i  \right) \right).
 \end{align*}
Hence, $\frac{d}{d\epsilon}\mathbb{P}_{\e}$ equals \eqref{c1-eq4}, given at the top of page~\pageref{c1-eq4}. Therefore, 
$\frac{d}{d\epsilon}\mathbb{P}_{\e}\left(0\right)$ equals \eqref{di}, and is a linear combination of $\zeta_i$ where $D_i$ equals \eqref{c1-n5}, given at the top of page~\pageref{c1-n5}. Let $Y_i\sim\mathsf{Poisson}\left(\lambda_i+d\right)$ be the output sequence when $B=1$ is transmitted. Then, for $0 \leq i \leq N+K-2$, we have
\newcounter{storeeqcounter3}
\newcounter{tempeqcounter3}
\setcounter{storeeqcounter3}{\value{equation}}
\addtocounter{equation}{1}
\begin{figure*}[!t]
\normalsize
\setcounter{tempeqcounter3}{\value{equation}}
\begin{align}
\setcounter{equation}{\value{storeeqcounter3}}
\label{c1-eq4}
\frac{d}{d\epsilon}\mathbb{P}_{\e}=
    \frac{1}{2}\sum_{\boldsymbol{y}:\sum a_j y_j < b} \left(  \left( \prod_{i=0}^{N+K-2}   e^{-\left(\lambda_i+d+\epsilon \zeta_i\right)} \frac{\left(\lambda_i+d+\epsilon \zeta_i\right)^{y_i}}{y_i !}   \right)  \left( \sum_{i=0}^{N+K-2} \left( \frac{y_i \zeta_i}{\lambda_i +d+\epsilon \zeta_i}-\zeta_i \right)  \right) \right)
\end{align}
\setcounter{equation}{\value{tempeqcounter3}} 
\hrulefill
\vspace*{4pt}
\end{figure*}
\newcounter{storeeqcounter4}
\newcounter{tempeqcounter4}
\setcounter{storeeqcounter4}{\value{equation}}
\addtocounter{equation}{1}
\begin{figure*}[!t]
\normalsize
\setcounter{tempeqcounter4}{\value{equation}}
\begin{align}
\setcounter{equation}{\value{storeeqcounter4}}
\label{di}
\frac{d}{d\epsilon}\mathbb{P}_{\e}\left(0\right) =\sum_{\boldsymbol{y}:\sum a_j y_j < b}  \left( \prod_{i=0}^{N+K-2} e^{-\left(\lambda_i+d\right)} \frac{{\left(\lambda_i+d\right)}^{y_i}}{y_i !} \right)   \left( \sum_{i=0}^{N+K-2} \zeta_i \frac{y_i -\left(\lambda_i+d\right)}{\left(\lambda_i+d\right) }  \right) =\sum_{i=0}^{N+K-2} D_i \zeta_i
\end{align}
\setcounter{equation}{\value{tempeqcounter4}} 
\hrulefill
\vspace*{4pt}
\end{figure*}
\newcounter{storeeqcountern5}
\newcounter{tempeqcountern5}
\setcounter{storeeqcountern5}{\value{equation}}
\addtocounter{equation}{1}
\begin{figure*}[!t]
\normalsize
\setcounter{tempeqcountern5}{\value{equation}}
\begin{align}
\setcounter{equation}{\value{storeeqcountern5}}
\label{c1-n5}
D_i =  &\sum_{\boldsymbol{y}:\sum a_j y_j < b}  \left( \prod_{j=0}^{N+K-2} e^{-\left(\lambda_j+d\right)} \frac{{\left(\lambda_j+d\right)}^{y_j}}{y_j !}  \right) 
 \left( \frac{y_i -\left(\lambda_i+d\right)}{\left(\lambda_i +d\right)} \right)
\end{align}
\setcounter{equation}{\value{tempeqcountern5}} 
\hrulefill
\vspace*{4pt}
\end{figure*}
 \begin{align}
 &D_i \nonumber\\ 
 &\ = \mathbb{E} \left[ \left(\frac{Y_i}{\lambda_i+d}-1\right) \boldsymbol{1} \left[\sum_{j=0}^{N+K-2} a_j Y_j < b\right] \right]\\
 &\ =\mathbb{P} \left[\sum_{j=0}^{N+K-2} a_j Y_j < b\right] \mathbb{E} \left[\frac{Y_i}{\lambda_i+d}-1  \middle|   \sum_{j=0}^{N+K-2} a_j Y_j < b\right]
\\& \ =\mathbb{P} \left[\sum_{j=0}^{N+K-2} a_j Y_j < b \right] \nonumber\\ 
& \ \quad \times \left( -1+\frac{1}{\lambda_i+d} \sum_{n=0}^{\infty} \mathbb{P} \left[Y_i \geq n  \middle|  \sum_{j=0}^{N+K-2} a_j Y_j < b \right]  \right),\label{le}
\end{align}
where \eqref{le} holds because for $X \geq 0$, ${\mathbb{E}[X] = \sum_n  P[X \geq n ]}$. 
This condition can be also expressed as follows:
\begin{align}
\label{too}
&\frac{d}{d\epsilon}\mathbb{P}_{\e}\left(0\right)\nonumber\\ 
& =\sum_{i=0}^{N+K-2} D_i \zeta_i \nonumber \\
& =  D_0 s_0 \pi_0 +D_1 \left(s_1 \pi_0 +s_0 \pi_1\right) \nonumber\\
&\quad +D_2 \left(s_2 \pi_0 +s_1 \pi_1 + s_0 \pi_2\right) \nonumber \\ \nonumber
&\quad +\ldots + D_{N+K-3} \left( s_{N-1}\pi_{K-2} +s_{N-2}\pi_{K-1}\right)\nonumber\\
&\quad +D_{N+K-2} \left( s_{N-1}\pi_{K-1} \right) \nonumber \\ \nonumber
& =   s_0 \left( D_0 \pi_0 + D_1 \pi_1+ D_2 \pi_2 + \ldots + D_{K-1} \pi_{K-1} \right) \\ \nonumber
&\quad + s_1 \left( D_1 \pi_0 + D_2 \pi_1+ D_3 \pi_2 + \ldots + D_{K} \pi_{K-1} \right)\\ 
&\quad + \ldots + s_{N-1} \left( D_{N-1} \pi_0 + D_{N} \pi_1+ \ldots + D_{N+K-2} \pi_{K-1} \right)
\end{align}
We begin by proving the second part of the lemma first.

\emph{Case 1:} If $x_{1j}=0$ for all $j$: in this case $s_j\geq 0$ for all $j$. Therefore, nonnegativity of the derivative of $\mathbb{P}_{\e}\left(0\right)$ in \eqref{too} implies that $\sum_{i=0}^{K-1} D_{j+i} \pi_i\geq 0$ for all $j$. This implies the second part of the lemma in this case.

\emph{Case 2:} There is a unique $j$ where $x_{1j}>0$:
 let
$s_{j'}=0$ for all $j'\neq j$. Setting $s_j\leq 0$, we obtain that $\sum_{i=0}^{K-1} D_{j+i} \pi_i\leq 0$. Let $\nu=\sum_{i=0}^{K-1} D_{j+i} \pi_i$. Consider an index $\tilde{j}$ where $x_{1\tilde j}=0$. Then, choosing $s_{\tilde j}\geq 0$ and $s_j=-s_{\tilde{j}}$ (and $s_i=0$ for $i\notin\{j,\tilde j\}$) we do not change the power of the codeword, and obtain  $\sum_{i=0}^{K-1} D_{\tilde j+i} \pi_i\geq \nu$ as a necessary condition. Next, if $x_{1j}$ (which is also the power of the codeword $\boldsymbol{x_1}$ in this case) is strictly less than $P$, we can also choose $s_j>0$. This shows that the coefficient $\sum_{i=0}^{K-1} D_{j+i} \pi_i= 0$. Also, similar to Case 1, we can conclude that $\sum_{i=0}^{K-1} D_{j'+i} \pi_i\geq 0$ for all $j'\neq j$. 

\emph{Case 3:} There exists two distinct indices $j$ and $j'$ where $x_{1j}>0$ and $x_{1j'}>0$. Using the choice of $s_j=-s_{j'}$ and $s_i=0$ for $i\notin\{j,j'\}$, we get that the coefficients of $s_j$ and $s_{j'}$ must be equal, \emph{i.e.,} $\sum_{i=0}^{K-1} D_{j+i} \pi_i=\sum_{i=0}^{K-1} D_{j'+i} \pi_i$. Thus, $\sum_{i=0}^{K-1} D_{j+i} \pi_i=\nu$ if  $x_{1j}>0$ for some constant $\nu$. Let us choose $s_j=s_{j'}$ and $s_i=0$ for $i\notin\{j,j'\}$. If the power of the codeword $\boldsymbol{x_1}$ is strictly less than the  total-power  budget $P$, we can choose $s_j=s_{j'}$ to be an arbitrary number; otherwise, we should set $s_j=s_{j'}\leq 0$. The first derivative condition then implies that $\nu\leq 0$ in general and $\nu=0$ if the power of the codeword $\boldsymbol{x_1}$ is strictly less than $P$. Finally, consider an index $\tilde{j}$ where $x_{1\tilde j}=0$. Then, as before choosing $s_{\tilde j}\geq 0$ and $s_j=-s_{\tilde{j}}$ (and $s_i=0$ for $i\notin\{j,\tilde j\}$) we do not change the power of the codeword, and obtain  $\sum_{i=0}^{K-1} D_{\tilde j+i} \pi_i\geq \nu$ as a necessary condition.

It remains to verify the first part of the lemma. 
Observe that
\begin{align*}
&\mathbb{P} \left[\sum_{j=0}^{N+K-2} a_j Y_j < b\right]  \left( -1+\frac{1}{\lambda_i+d} \sum_{n=0}^{\infty} \mathbb{P}\left[Y_i \geq n\right]  \right) \\
&\ = \mathbb{P} \left[\sum_{j=0}^{N+K-2} a_j Y_j < b\right]  \left( -1+\frac{1}{\lambda_i+d} \mathbb{E}[Y_i]  \right) 
\\& \ =0.
\end{align*}
Thus, we have $D_i$ equals \eqref{c1-eq5} that is given at the top of page~\pageref{c1-eq5}.
\newcounter{storeeqcounter5}
\newcounter{tempeqcounter5}
\setcounter{storeeqcounter5}{\value{equation}}
\addtocounter{equation}{1}
\begin{figure*}[!t]
\normalsize
\setcounter{tempeqcounter5}{\value{equation}}
\begin{align}
\setcounter{equation}{\value{storeeqcounter5}}
\label{c1-eq5}
D_i &= \mathbb{P} \left[\sum_{j=0}^{N+K-2} a_j Y_j < b\right]  \left( -1+\frac{1}{\lambda_i+d} \sum_{n=0}^\infty \left(\mathbb{P}\left[Y_i \geq n  \middle|  \sum_{j=0}^{N+K-2} a_j Y_j < b \right] - \mathbb{P}[Y_i \geq n]
\right)
 \right)
\end{align}
\setcounter{equation}{\value{tempeqcounter5}} 
\hrulefill
\vspace*{4pt}
\end{figure*}

We show that for any $n$ we have $$\mathbb{P}[Y_i \geq n \vert \sum_{j=0}^{N+K-2} a_j Y_j < b ] - \mathbb{P}[Y_i \geq n]\leq 0,$$ if and only if  $\lambda_i \geq \mu_i$. 
We have $\lambda_i \geq \mu_i$ if and only if $a_i \geq 0$. Since $Y_i$ are mutually independent, it suffices to show that for any arbitrary values $\boldsymbol{y}_{\sim i}=\left(y_1, y_2,  \ldots , y_{i-1}, y_{i+1},  \ldots , y_n\right)$ we have
$$\mathbb{P}\left[Y_i \geq n  \middle|   \sum_{j=0}^{N+K-2} a_j Y_j < b, \boldsymbol{Y}_{\sim i}=\boldsymbol{y}_{\sim i} \right] - \mathbb{P}\left[Y_i \geq  n\right]  \leq  0$$
if and only if $a_i \geq 0$. In other words, we need to show that the event $Y_i \geq n$ is negatively correlated with $a_iY_i<b- \sum_{j\neq i} a_j y_j$ if and only if $a_i \geq 0$.    This follows from Corollary~\ref{fkgc} for functions $\mathbf{1}[y_i \geq n]$ and $ \mathbf{1}[a_iy_i< b- \sum_{j\neq i} a_j y_j]$. With the same argument we can show that $\lambda_i \leq \mu_i $ if and only if $D_i \geq 0$. Therefore $\lambda_i=\mu_i$ if and only if $D_i=0$.  
\end{proof}

\begin{lemma} 
\label{fkg}
Let $X: \Omega \to \mathbb{Z}$ be a discrete random variable with probability mass function $p\left(x\right)$, and $f, g : \mathbb{Z} \to \mathbb{Z}$ 
be two increasing functions such that $\mathbb{E}[f(X)^2] < \infty$ and $\mathbb{E}[g(X)^2] < \infty$, then we have:
\begin{align}
\mathbb{E}[f(X)g(X)] \geq \mathbb{E}[f(X)]\mathbb{E}[g(X)].
\end{align} 
\end{lemma}
\begin{proof} The proof follows a technique that is used in the proof of the FKG inequality, e.g., see \cite[Sec. 2.2]{Grimmett}.
Suppose that $X_1$ and $X_2$ are two independent random variables having the same distribution as $X$, then we have:
\begin{align}
\label{trick}
\left(f(X_1)-f(X_2)\right)\left(g(X_1)-g(X_2)\right)\geq 0.
\end{align}
Since $f, g $ have finite second moments, the expectation of the left hand side of the above inequality exists (using the Cauchy–Schwarz inequality), therefore we have:
\begin{align}
&\mathbb{E}_{X_2}\mathbb{E}_{X_1}[\left(f(X_1)-f(X_2)\right)\left(g(X_1)-g(X_2)\right)] \geq 0.
\end{align}
After expanding the above inequality, we have:
\begin{align}
\label{e1111}
&\mathbb{E}[f(X_1)g(X_1)] + \mathbb{E}[f(X_2)g(X_2)] \nonumber\\
& \ - \mathbb{E}[f(X_2)]\mathbb{E}[g(X_1)] - \mathbb{E}[f(X_1)]\mathbb{E}[g(X_2)] \geq 0.
\end{align}
Because $X_1$ and $X_2$ have the same distribution as $X$, \eqref{e1111} results in:
\begin{align*}
\mathbb{E}[f(X)g(X)] \geq \mathbb{E}[f(X)]\mathbb{E}[g(X)].
&\qedhere
\end{align*}
\end{proof}
\begin{corollary}
\label{fkgc}
The same inequality is true for two decreasing functions $f,g$ (have finite second moments) by applying Lemma~\ref{fkg} to $-f, -g$. Also, if one of them (for example $f$) is increasing and the other one ($g$) is decreasing by applying Lemma~\ref{fkg} to $f$ and $-g$ we have:
   \begin{align}
   \mathbb{E}[f(X)g(X)] \leq \mathbb{E}[f(X)]\mathbb{E}[g(X)].
\end{align}
\end{corollary}

\bibliographystyle{IEEEtran}
\bibliography{IEEEabrv,bibtex}

\begin{IEEEbiographynophoto}{Niloufar Ahmadypour}
 received the B.Sc. degree in Electrical Engineering and Pure Mathematics and the M.Sc. degree in Electrical Engineering all from the Isfahan University of Technology, Isfahan, Iran, in 2013, 2014 and 2015 respectively. She is currently pursuing the Ph.D. degree in Electrical Engineering at the Sharif University of Technology,
Tehran, Iran. She received the Third Prize and The Second Prize from the 20th and 21st International Mathematical Competition
for University Students (IMC 2013 and IMC 2014). Her research interests include various topics in information theory, theoretical computer science, and high-dimensional statistics.
\end{IEEEbiographynophoto}

\begin{IEEEbiographynophoto}{Amin Gohari}
 received the B.Sc. degree from Sharif University of Technology,
Tehran, Iran, in 2004 and his Ph.D. degree in
Electrical Engineering from the University
of California, Berkeley in 2010. From 2010-
2011, he was a postdoc at the Chinese University of Hong Kong, Institute of Network
Coding. From 2011-2020 he was with the
Electrical Engineering department of Sharif
University of Technology. He joined the Tehran Institute for Advanced Studies in 2020. Dr. Gohari received the 2010 Eli Jury
Award from UC Berkeley, Department of Electrical Engineering
and Computer Sciences, for outstanding achievement in the area
of communication networks, and the 2009-2010 Bernard Friedman
Memorial Prize in Applied Mathematics from UC Berkeley, Department of Mathematics, for demonstrated ability to do research
in applied mathematics. He also received the Gold Medal from
the 41st International Mathematical Olympiad (IMO 2000) and the
First Prize from the 9th International Mathematical Competition
for University Students (IMC 2002). He was a finalist for the
best student paper award at IEEE International Symposium on
Information Theory (ISIT) in three consecutive years, 2008, 2009
and 2010. He was also a co-author of a paper that won the ISIT
2013 Jack Wolf student paper award, and two that were finalists
in 2012 and 2014. He was selected as an exemplary reviewer for
Transactions on Communications in 2016. Dr. Gohari is currently
serving as an Associate Editor for the IEEE Transactions on
Information Theory.
\end{IEEEbiographynophoto}
\end{document}